%% file: main.tex
\documentclass[letter,11pt]{article}

\usepackage{algorithm}
\usepackage{amsmath}
\usepackage{amsthm}
\usepackage[noend]{algpseudocode}
\usepackage{enumitem}
\usepackage{geometry}
\usepackage{graphicx}
\usepackage{listings}
\usepackage{mathtools}
\usepackage{subcaption}
\usepackage{tabularx}
\usepackage{tcolorbox}
\usepackage{varwidth}
\usepackage{stmaryrd}

\newcommand{\myparagraph}[1]{\smallskip\noindent \textbf{\textit{#1. }}}

\newtheorem{theorem}{Theorem}
\newtheorem{lemma}{Lemma}
\newtheorem{definition}{Definition}
\newtheorem{corollary}{Corollary}

\newtheorem{invariant}{Invariant}
\newtheorem{property}{Property}

\newcommand{\algtype}[1]{{\itshape\bfseries{#1}}}

\newcommand{\code}[1]{\lstinline!#1!}

\newcommand{\ra}[2]{#1.#2}
\newcommand{\raa}[2]{#1.#2}

\algrenewcommand\algorithmicindent{1em}

\algnewcommand{\algcomment}[1]{\qquad{\color{blue}\emph{// #1}}}    
\algnewcommand{\alglinecomment}[1]{{\color{purple!40!black}\emph{// #1}}}      

\algnewcommand\alglocal{\textbf{local}}                            
\algnewcommand\algreturn{\textbf{return}}                          
\algnewcommand\algeach{\textbf{each}}                              
\algnewcommand\algretire{\textbf{retire process}}                   

\algnewcommand\algmodassign{\textbf{write}}                         
\algnewcommand\algwritemod[2]{\algmodassign(#1, #2)}          		
\algnewcommand\algto{\textbf{to}}									
\algnewcommand\algis{\textbf{is}}                                   
\algnewcommand\algnot{\textbf{not}}                                 
\algnewcommand\algin{\textbf{in}}                                   
\algnewcommand\algempty{\textbf{empty}}                             

\algnewcommand\algand{\textbf{and}}                               
\algnewcommand\algor{\textbf{or}}                                 
\algnewcommand\algassign{\ensuremath{\gets}}                       

\algnewcommand\algtrue{\textbf{true}}                               
\algnewcommand\algfalse{\textbf{false}}                             
\algnewcommand\algnull{\ensuremath{\perp}}                          

\algnewcommand\algarray[1]{\textnormal{array}\ensuremath{\langle}#1\ensuremath{\rangle}}    
\algnewcommand\algmod[1]{\textbf{mod}\ensuremath{\langle}#1\ensuremath{\rangle}}            

\algnewcommand\algorithmicwith{\textbf{with}}
\algnewcommand\algorithmicread{\textbf{read}}
\algnewcommand\algorithmicas{\textbf{as}}
\algnewcommand\Read[2]{\State \alglocal #2 \algassign \algorithmicread(#1)}
\algnewcommand\EndRead{}

\algnewcommand\algorithmicinparallel{\textbf{in parallel}}
\algblockdefx[PARFOR]{ParallelFor}{EndParallelFor}[1]
{\algorithmicfor\ #1\ \algorithmicdo\ \algorithmicinparallel}
{\algorithmicend\ \algorithmicfor}

\makeatletter
\ifthenelse{\equal{\ALG@noend}{t}}%
{\algtext*{EndParallelFor}}
{}%
\makeatother

\mathtoolsset{showonlyrefs}

\newenvironment{proofsketch}{%
  \proof}{\endproof}

\DeclareMathOperator{\polylog}{polylog}

\newcommand{\multipram}[0]{\emph{multiprefix} CRCW \ensuremath{\mathsf{PRAM}}\renewcommand{\multipram}[0]{multiprefix CRCW \ensuremath{\mathsf{PRAM}}}}

\newcommand{\g}[1]{{\color{red} #1}}

\begin{document}

  \title{Parallel Minimum Cuts in $O(m \log^2(n))$ Work and Low Depth\footnote{This is the full version of the paper appearing in the ACM Symposium on Parallelism in Algorithms and Architectures (SPAA), 2021}}

  \author{Daniel Anderson \\ Carnegie Mellon University \\ dlanders@cs.cmu.edu \and  Guy E. Blelloch \\  Carnegie Mellon University \\  guyb@cs.cmu.edu}

  \date{}
  \maketitle
  
  \input{abstract.tex}

  \pagenumbering{gobble}
  
  \clearpage
  
  \pagenumbering{arabic}
  
  \input{introduction.tex}

  \input{preliminaries.tex}


\input{mixed.tex}
  
  \input{packing.tex}

  \input{two-respecting.tex}

\input{conclusion.tex}

  {\footnotesize
  \bibliographystyle{abbrv}
  \bibliography{ref}
  }
  
  
  
  
  
  
  

  
\end{document}

%% file: abstract.tex
\begin{abstract}
We present a randomized $O(m \log^2 n)$ work, $O(\polylog n)$ depth parallel algorithm for minimum cut. This algorithm matches the work bounds of a recent sequential algorithm by Gawrychowski, Mozes, and Weimann [ICALP'20], and improves on the previously best parallel algorithm by Geissmann and Gianinazzi [SPAA'18], which performs $O(m \log^4 n)$ work in $O(\polylog n)$ depth.

Our algorithm makes use of three components that might be of independent interest. Firstly, we design a parallel data structure that efficiently supports batched mixed queries and updates on trees. It generalizes and improves the work bounds of a previous data structure of Geissmann and Gianinazzi and is work efficient with respect to the best sequential algorithm. Secondly, we design a parallel algorithm for approximate minimum cut that improves on previous results by Karger and Motwani. We use this algorithm to give a work-efficient procedure to produce a tree packing, as in Karger's sequential algorithm for minimum cuts. Lastly, we design an efficient parallel algorithm for solving the minimum $2$-respecting cut problem.
\end{abstract}

%% file: introduction.tex
\section{Introduction}

Minimum cut is a classic problem in graph theory and algorithms. The problem is to find, given an undirected weighted graph $G = (V, E)$, a nonempty subset of vertices $S \subset V$ such that the total weight of the edges crossing from $S$ to $V \setminus S$ is minimized. Early approaches to the problem were based on reductions to maximum $s$-$t$ flows~\cite{gomory1961multi,hao1994faster}. Several algorithms followed which were based on edge contraction~\cite{nagamochi1992computing,nagamochi1992linear,karger1993global,karger1996new}. Karger was the first to observe that tree packings~\cite{nash1961edge} can be used to find minimum cuts~\cite{karger2000minimum}.
In particular, for a graph with $n$ vertices and $m$ edges, Karger showed how to use random sampling and a tree packing algorithm of Gabow~\cite{gabow1995matroid} to generate a set of $O(\log n)$ spanning trees such that, with high probability, the minimum cut crosses at most two edges of one of them. A cut that crosses at most $k$ edges of a given tree is called a \emph{$k$-respecting} cut. Karger then gives an $O(m \log^2 n)$-time algorithm for finding minimum 2-respecting cuts, yielding a randomized $O(m\log^3 n)$-time algorithm for minimum cut.  Karger also gives a parallel algorithm for minimum 2-respecting cuts in $O(n^2)$ work and $O(\log^3 n)$ depth.

Until very recently, these were the state-of-the-art sequential and parallel algorithms for the weighted minimum cut problem. A new wave of interest in the problem has recently pushed these frontiers. Geissmann and Gianinazzi~\cite{geissmann2018parallel} design a parallel algorithm for minimum $2$-respecting cuts that performs $O(m\log^3 n)$ work in $O(\log^2 n)$ depth. Their algorithm is based on parallelizing Karger's algorithm
  by replacing a sequential data structure for the so-called \emph{minimum path} problem, based on dynamic trees, with a data structure that can evaluate a \emph{batch} of updates and queries in
  parallel.  Their algorithm performs just a factor of $O(\log n)$ more work than Karger's sequential algorithm, but substantially improves on the work of Karger's parallel algorithm.
  
Soon after, a breakthrough from Gawrychowski, Mozes, and Weimann~\cite{gawrychowski2019minimum} gave a randomized $O(m\log^2 n)$ algorithm for minimum cut. Their algorithm achieves the  $O(\log n)$ speedup by designing an $O(m\log n)$ algorithm for finding the minimum $2$-respecting cuts, which was the bottleneck of Karger's algorithm. This is the first result to beat Karger's seminal algorithm in over 20 years.

An open question posed by Karger was whether a deterministic algorithm can achieve an $O\left(m^{1+o(1)}\right)$ runtime. This was recently resolved in the affirmative by Li~\cite{li2020deterministic} by derandomizing the construction of the spanning trees.


  In our work, we combine ideas from Gawrychowski et al.\ and Geissmann and Gianinazzi with several new techniques to close the gap between
  the parallel and sequential algorithms.   Our contribution can be summarized by:
  
  \begin{theorem}\label{thm:main-theorem}
    The minimum cut of a weighted graph can be computed with high probability in $O(m \log^2 n)$ work and $O(\log^3 n)$ depth.
  \end{theorem}
  
  \noindent We achieve this using a combination of results that may be of independent interest. Firstly, we design a framework for evaluating mixed batches of updates and queries on trees work efficiently in low depth. This algorithm is based on parallel Rake-Compress Trees (RC trees)~\cite{acar2020batch}. Roughly, we say that a set of update and query operations implemented on an RC tree is \emph{simple} (defined formally in Section~\ref{sec:batched-mixed-ops}) if the updates maintain values at the leaves that are modified by an associative operation and combined at the internal nodes, and the queries read only the nodes on a root-to-leaf path and their children. Simple operation sets include updates and queries on path and subtree weights.
  
  \begin{theorem}\label{theorem:batch}
    Given a bounded-degree RC tree of size $n$ and a simple operation set, after $O(n)$ work and $O(\log n)$ depth preprocessing, batches of $k$ operations from the operation-set, can be processed in
    $O(k \log (kn))$ work and $O(\log n\log k)$ depth. The total space required is $O(n + k_{max})$, where $k_{max}$ is the maximum size of a batch.
  \end{theorem}
  
  \noindent This result generalizes and improves on Geissmann and Gianinazzi~\cite{geissmann2018parallel} who give an algorithm for evaluating a batch of $k$ path-weight updates and queries in $\Omega(k\log^2 n)$ work.
  
  Next, we design a faster parallel algorithm for approximating minimum cuts, which is used as an ingredient in producing the tree packing used in Karger's approach (Section~\ref{sec:packing}). To achieve this, we design a faster sampling scheme for producing graph skeletons, leveraging recent results on sampling binomial random variables, and a transformation that reduces the maximum edge weight of the graph to $O(m \log n)$ while approximately preserving cuts.
 
  Lastly, we show how to solve the minimum $2$-respecting cut problem efficiently in parallel, using a combination of our new mixed batch tree operations algorithm and the use of RC trees to efficiently perform a divide-and-conquer search over the edges of the $2$-constraining trees (Section~\ref{sec:two-respect}).
  
  \begin{theorem}\label{thm:two-respecting}
    The minimum $2$-respecting cut of a weighted graph with respect to a given spanning tree can be computed in $O(m \log n)$ work and $O(\log^3 n)$ depth with high probability.
  \end{theorem}

\myparagraph{Application to the unweighted problem}
The unweighted minimum cut problem, or edge connectivity problem was recently improved
by Ghafarri, Nowicki, and Thorup~\cite{ghaffari2020faster} who give an $O(m\log n + n\log^4 n)$
work and $O(\polylog n)$ depth randomized algorithm which uses Geissmann and Gianinazzi's algorithm as a subroutine. By plugging
our improved algorithm into Ghafarri, Nowicki, and Thorup's algorithm, we obtain an
algorithm that runs in $O(m\log n + n\log^2 n)$ work
and $O(\polylog n)$ depth w.h.p.

%% file: preliminaries.tex
\section{Preliminaries}
\label{sec:prelims}
  
  \myparagraph{Model of computation} 
  We analyze algorithms in the \emph{work-depth} model using fork-join parallelism. A procedure can \emph{fork} another procedure call to run in parallel and then wait for forked procedures to complete with a \emph{join}. Work is defined as the total number of instructions performed by the algorithm and depth (also called span) is the length of the longest chain of sequentially dependent instructions \cite{Blelloch96}. The model can work-efficiently cross simulate the classic CRCW PRAM model~\cite{Blelloch96}, and the more recent Binary Forking model~\cite{blelloch2020optimal} with at most a logarithmic-factor difference in the depth.
  
  \myparagraph{Randomness}
  We say that a statement happens \emph{with high probability} (w.h.p) in $n$ if for any constant $c$, the constants in the statement can be set such that the probability that the event fails to hold is $O(n^{-c})$.
  In line with Karger's work on random sampling~\cite{karger1999random}, we assume that we can generate $O(1)$ random bits in $O(1)$ time. Since some of the subroutines we use require random $\Theta(\log n)$-bit words, these take $O(\log n)$ work to generate.  The depth is unaffected since we can always pre-generate the anticipated number of random words at the beginning of our algorithms.
  
  Our algorithms are Monte Carlo, i.e.,\ correct w.h.p.\ but run in a deterministic amount of time. We can use Las Vegas algorithms, which are fast w.h.p.\ but always correct, as subroutines, because any Las Vegas algorithm can be converted into a Monte Carlo algorithm by halting and returning an arbitrary answer after the desired time.

  \myparagraph{Tree contraction}
  Parallel tree contraction is a technique developed to efficiently apply various operations over trees in logarithmic parallel depth~\cite{miller1989parallel}, and was also later applied to dynamic trees~\cite{acar2005experimental}.   
Tree contration consists of a set of rake and compress operations.
  The \emph{rake} operation removes a leaf vertex and merges it with its parent.
  The \emph{compress} operation removes a vertex of degree two and replaces its two incident edges with a single edge
  joining its neighbors. Miller and Reif~\cite{miller1989parallel} observed that rakes and compresses can be applied in parallel as long as they
are applied to an independent set of vertices.   They describe a random-mate technique that ensures that 
any tree contracts to a single vertex in $O(\log n)$ rounds w.h.p., and using a total of $O(n)$ work in expectation.
Gazit, Miller, and Teng~\cite{gazit1988optimal} give a deterministic version with the same bounds, and Blelloch et al.~\cite{blelloch2020optimal} give a version that works in the binary-forking model. Miller and Reif's algorithm applies to bounded-degree trees, but arbitrary-degree
trees can typically be handled by converting them into bounded-degree trees. For a rooted tree, the root is never removed, and is the final surviving vertex.



  \myparagraph{Rake-compress trees} The RC tree~\cite{acar2005experimental,acar2020batch} of a tree $T$ encodes a recursive clustering of $T$ corresponding to the result of tree contraction, where each cluster corresponds to a rake or compress (see Figure~\ref{fig:rc-tree}).  A cluster is defined to be a connected subset of vertices and edges of the original
  tree. Importantly, a cluster can contain an edge without containing
  its endpoints. The \emph{boundary vertices} of a cluster $C$ are the vertices $v \notin C$
  such that an edge $e \in C$ has $v$ as one of its endpoints. All of the clusters
  in an RC tree have at most two boundary vertices. A cluster with no boundary vertices
  is called a \emph{nullary cluster} (generated at the top-level \emph{root} cluster), a cluster with one boundary is a
  \emph{unary cluster} (generated by the rake operation)
  and a cluster with two boundaries is \emph{binary cluster}
  (generated by the compress operation). The \emph{cluster path} of a binary cluster is the path in $T$ between its boundary vertices. Nodes in an RC tree correspond to clusters, such that a node is the disjoint union of its children.
  
  The leaf clusters of the RC tree are the vertices and
  edges of the original tree, which are nullary and binary clusters respectively. Note that all non-leaf clusters have exactly one vertex (leaf)
  cluster as a child. This vertex is that cluster's \emph{representative} vertex.
  The recursive clustering is then defined by the following simple rule: Each rake or compress operation corresponds to a cluster, such that the operation that deletes vertex $v$ from the tree defines a cluster with representative vertex $v$ whose non-leaf subclusters are all of the clusters that have $v$ as a boundary vertex. Clusters therefore have the useful property that the constituent clusters of a parent cluster $C$ share a single boundary vertex in common---the representative of $C$, and their remaining boundary vertices become the boundary vertices of $C$.

In this paper we will be considering rooted trees.  In this case the root of the tree is also the representative of the top level nullary cluster of the RC-tree, e.g. vertex \texttt{e} in Figure~\ref{fig:rc-tree}.  Non-leaf binary clusters have a binary subcluster whose cluster path is above the representative vertex in the input tree, which we will refer to as the \emph{top cluster}, and a binary subcluster whose cluster path is below the representative vertex, which we call the \emph{bottom cluster}.  We will also refer to the binary subcluster of a unary cluster as the top cluster as its cluster path is also above the representative vertex.
In our pseudocode, we will use the following notation. For a cluster $x$: $\ra{x}{v}$ is the representative vertex, $\ra{x}{t}$ is the top subcluster, $\ra{x}{b}$ is the bottom subcluster,  $\ra{x}{U}$ is a list of unary subclusters, and $\ra{x}{p}$ is the parent cluster.


\begin{figure}[t]
  \centering
  \begin{subfigure}{0.4\columnwidth}
    \centering
    \includegraphics[width=0.9\columnwidth]{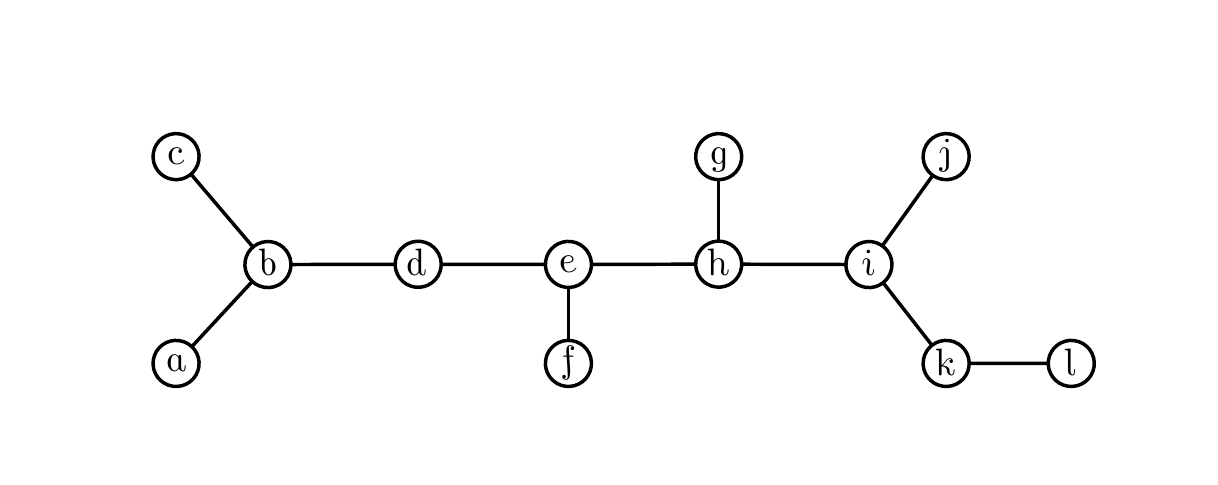}
    \caption{A tree}
  \end{subfigure}
  \begin{subfigure}{0.4\columnwidth}
    \centering
    \includegraphics[width=0.9\columnwidth]{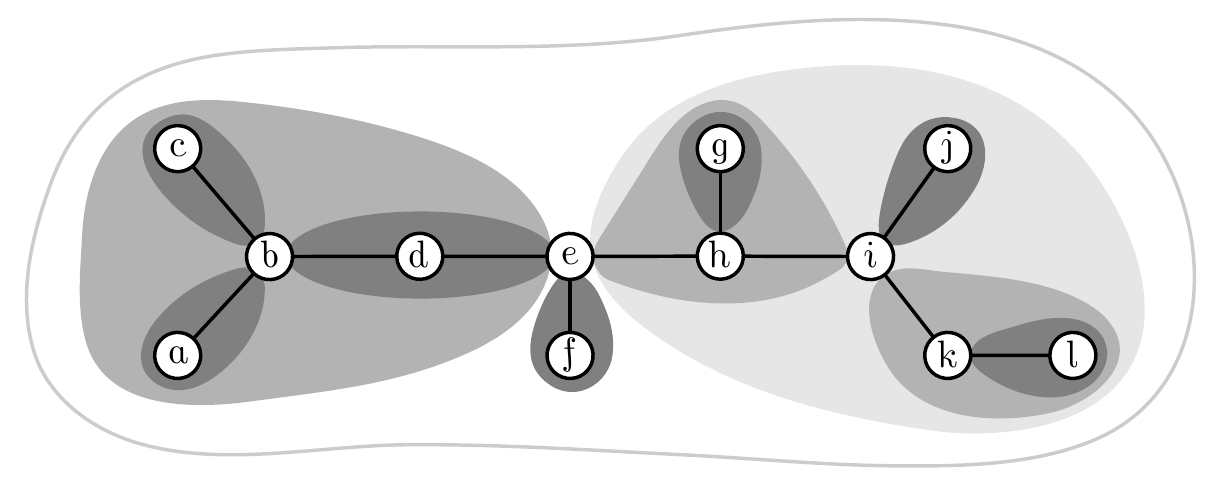}
    \caption{A recursive clustering of the tree produced by tree contraction. Clusters produced in earlier rounds are depicted in a darker color.}
  \end{subfigure}
  \begin{subfigure}{0.8\columnwidth}
    \bigskip  
      \centering
      \includegraphics[width=0.8\columnwidth]{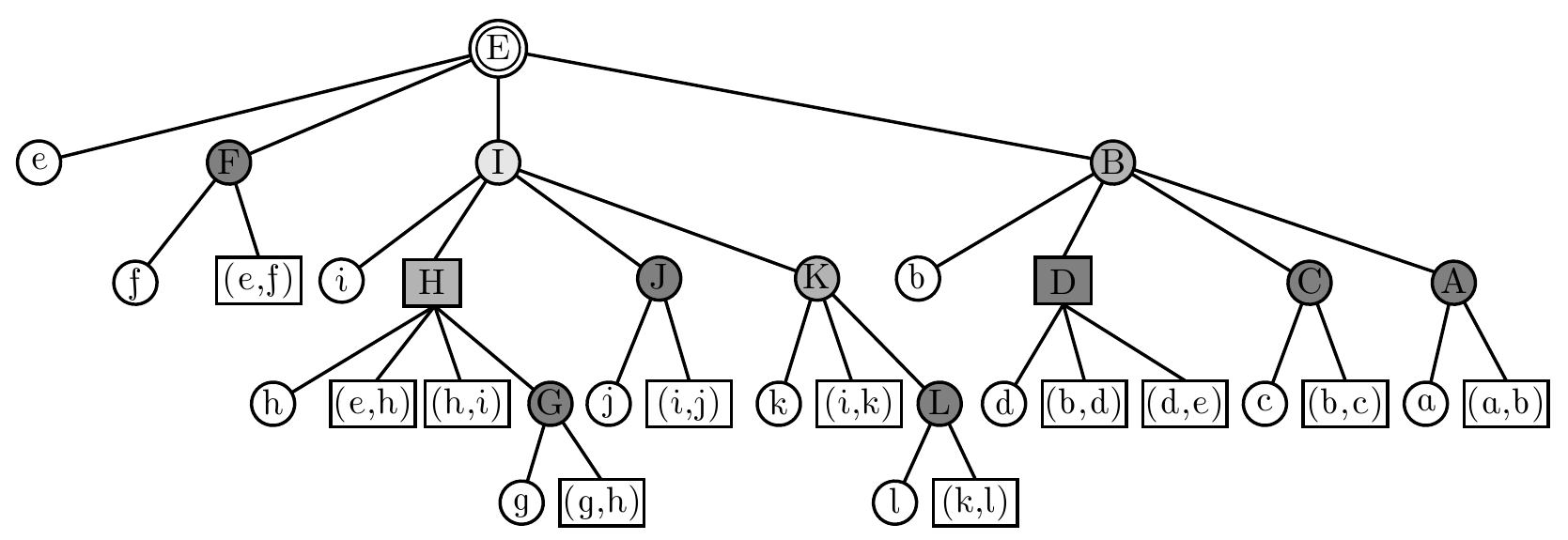}
      \caption{The corresponding RC tree. Unary clusters (from rakes) are shown as filled circles, binary clusters as rectangles, and the finalize (nullary) cluster at the root with two concentric circles. The leaf clusters are labeled in lowercase, and the composite clusters are labeled with the uppercase of their representative. The shade of a cluster corresponds to its height in the clustering. Lower heights (i.e., contracted earlier) are darker.}
    \end{subfigure}  
    \caption{A tree, a clustering, and the corresponding RC tree \cite{acar2020batch}.}\label{fig:rc-tree}
  \end{figure}

  \myparagraph{Compressed path trees} For a weighted (unrooted) tree $T$ and a set of \emph{marked} vertices $V \subset V(T)$, the compressed path tree is a weighted tree $T_c$ on some subset of the vertices of $T$ including $V$ with the following property: for every pair of vertices $(u,v) \in V\times V$, the weight of the lightest edge on the path from $u$ to $v$ is the same in $T$ and $T_c$. The compressed path three $T_c$ is defined as the smallest such tree. Alternatively, the compressed path tree is the tree $T$ with all unmarked vertices of degree less than three spliced out, where each spliced-out path is replaced by an edge whose weight is the lightest of the weights on the path it replaced. It is not hard to show that $T_c$ has size less than $2|V|$.
  Compressed path trees are described in \cite{anderson2020work}, where it is shown that given an RC tree for the tree $T$ and a set of $k$ marked vertices, the compressed path tree can be produced in $O(k\log(1+n/k))$ work and $O(\log^2 n)$ depth w.h.p.
   Gawrychowski et al.~\cite{gawrychowski2019minimum}  define a similar notion which they call ``topologically induced trees'', but their algorithm is sequential and requires $O(k \log n)$ work (time).
    
    \myparagraph{Karger's minimum cut algorithm} Karger's algorithm for minimum cuts~\cite{karger2000minimum} is based on the notion of \emph{$k$-respecting cuts}.
    Karger's algorithm is the following two-step process.
    \begin{enumerate}[leftmargin=*]
    	\item Find $O(\log n)$ spanning trees of $G$ such that w.h.p., the minimum cut $2$-respects at least one of them
    	\item Find, for each of the aforementioned spanning trees, the minimum $2$-respecting cut in $G$
    \end{enumerate}
	Karger solves the first step using a combination of random sampling and \emph{tree packing}. Given a weighted graph $G$, a tree packing of $G$ is a set of weighted spanning trees of $G$ such that for each edge in $G$, its total weight across
  all of the spanning trees is no more than its weight in $G$. The underlying tree packing algorithms used by Karger have running time proportional to the size of the minimum cut, so random sampling is first used to produce a sparsified graph, or \emph{skeleton}, where the minimum cut has size $\Theta(\log n)$ w.h.p. The sampling process is carefully crafted such that the resulting tree packing still has the desired property w.h.p.
  
	Given the skeleton graph, Karger gives two algorithms for producing tree packings such that sampling $\Theta(\log n)$ trees from them guarantees that, w.h.p., the minimum cut $2$-respects one of them. The first approach uses a tree packing algorithm of Gabow~\cite{gabow1995matroid}. The second is based on the packing algorithm of Plotkin et al.~\cite{plotkin1995fast}, and is much more amenable to parallelism. It works by performing $O(\log^2 n)$ minimum spanning tree computations. In total, Step 1 of the algorithm takes $O(m + n\log^3 n)$ time.
	
	For the second step, Karger develops an algorithm to find, given a graph $G$ and a spanning tree $T$, the minimum cut of $G$ that $2$-respects $T$. The algorithm works by arbitrarily rooting the tree, and considering two cases: when the two cut edges are on the same root-to-leaf path, and when they are not. Both cases use a similar technique; They consider each edge $e$ in the tree and try to find the best matching $e'$ to minimize the weight of the cut induced by the edges $\{ e, e' \}$. This is achieved by using a dynamic tree data structure to maintain, for each candidate $e'$, the value that the cut would have if $e'$ were selected as the second cutting edge, while iterating over the possibilities of $e$ and updating the dynamic tree. Karger shows that this step can be implemented sequentially in $O(m\log^2 n)$ time, which results in a total runtime of $O(m \log^3 n)$ when applied to the $O(\log n)$ spanning trees.

%% file: mixed.tex
\newcommand{\simpleRC}{simple RC}
\newcommand{\la}{\leftarrow}
\newcommand{\addPath}{\textproc{AddPath}}
\newcommand{\querySubtree}{\textproc{QuerySubtree}}
\newcommand{\queryPath}{\textproc{QueryPath}}
\newcommand{\queryEdge}{\textproc{QueryEdge}}

\section{Batched Mixed Operations on Trees}\label{sec:batched-mixed-ops}

The batched mixed operation problem is to take an off-line sequence of mixed operations on a data structure, usually a mix of queries and updates, and process them as a batch.  The primary reason for batch processing is to allow for parallelism on what would otherwise be a sequential execution of the operations.   We use the term \emph{operation-set} to refer to the set of operations that can be applied among the mixed operations. We are interested in operations on trees, and our results apply to operation-sets that can be implemented on an RC tree in a particular way, defined as follows.

\begin{definition}
An implementation of an operation-set on trees is a \emph{\simpleRC{} implementation} if it uses an RC representation
of the trees and satisfies the following conditions.
\begin{enumerate}[leftmargin=*]
\setlength\itemsep{0em}
\item 
The implementation maintains a value at every RC cluster that can be calculated in constant time from the values of the children of the cluster,
\item 
every query operation is implemented by traversing from a leaf to the root examining values at the visited clusters and their children taking contant time per value examined, and using constant space, and
\item
every update operation involves updating the value of a leaf using an associative constant-time operation, and then reevaluating the values on each cluster on the path from the leaf to the root.
\end{enumerate}
\end{definition} 
\noindent Note that every operation has an \emph{associated leaf} (either an edge or vertex).  Also note that setting (i.e.,\ overwriting) a value is an associative operation (just return the second of the arguments).  For \simpleRC{} implementations, all operations take time (work) proportional to the depth of the RC tree since they only follow a path to the root taking constant time at each cluster.  Although the \simpleRC{} restriction may seem contrived, most operations on trees studied in previous work~\cite{sleator1983data,alstrup2005maintaining,acar2005experimental} can be implemented in this form, including most path and subtree operations. This is because of a useful property of RC trees, that all paths and subtrees in the source tree can be decomposed into clusters that are children of a single path in the RC tree, and typically operations need just update or collect a contribution from each such cluster.

\myparagraph{Example}
As an example, consider the following two operations on a rooted tree (the first an update, and the second a query): 
\begin{itemize}[leftmargin=*]
\setlength\itemsep{0em}
\item \textproc{addWeight}$(v, w)$ : adds weight $w$ to a vertex $v$ 
\item \textproc{subtreeSum}$(v)$ : returns the sum of the weights of all of the vertices in the subtree rooted at $v$
\end{itemize}

  \begin{algorithm}[H]
    \begin{algorithmic}[1]
    \footnotesize
\Procedure{subtreeSum}{$v$ : \textbf{vertex}}
  \State $w \la 0$
  \State $x \la v$;  $p \la \ra{x}{p}$
  \While{$p$ \textbf{is a} binary cluster}
     \If{$(x = \ra{p}{t})$ or $(x = \ra{p}{v})$\label{line:subsumif}}
          \State $w  \la  w + \ra{p}{\ra{b}{w}} + \ra{p}{\ra{v}{w}} + \sum_{u \in \raa{p}{U}} \ra{u}{w}$ 
     \EndIf
     \State $x \la p;$ $p \la \ra{x}{p}$
  \EndWhile
  \State return $w + \ra{p}{\ra{v}{w}} + \sum_{u \in \raa{p}{U}} \ra{u}{w} $
\EndProcedure
\end{algorithmic}
\caption{The \textproc{subtreeSum} query.}
\label{alg:subsum}
\end{algorithm}

\noindent These operations can use a \simpleRC{} implementation by keeping as the value of each cluster the sum of values of all its children.  This satisfies the first condition since the sums take constant time. Single-edge clusters in the RC tree start with the initial weight of the edge, while single-vertex clusters start with zero weight.  An \textproc{addWeight}$(v,w)$ adds weight $w$ to the vertex $v$ (which is a leaf in the RC tree) and updates the sums up to the root cluster.  This satisfies the third condition since addition is associative and takes constant time.  The query can be implemented as in Algorithm~\ref{alg:subsum}, where $\ra{x}{w}$ is the weight stored on the cluster $x$.
It starts at the leaf for $v$ and goes up the RC tree keeping track of the total weight underneath $v$. Note that $x$ will never be a unary cluster, so if not the representative or top subcluster of $p$, it is the bottom subcluster with nothing below it in this cluster. Observe that \textproc{subtreeSum} only examines values on a path from the
start vertex to the root and the children along that path.   Each step takes constant time and requires constant space,
satisfying the second condition.   The operation-set therefore has a \simpleRC{} implementation.

\begin{figure*}[t]
\centering
\includegraphics[width=0.8\textwidth]{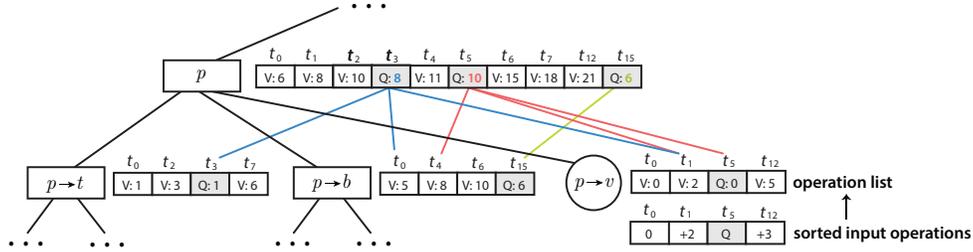}
\caption{Merging the operation lists for a binary cluster consisting of \textproc{addWeight} and \textproc{subtreeSum} operations. Values in the operation sequence, denoted V : $v$, are computed by aggregating the latest values of the children at the given timestamp. For example, at $t_6$ in $p$, the algorithm adds $3$ from $\ra{p}{t}$ at $t_2$, $10$ from $\ra{p}{b}$ at $t_6$, and $2$ from $\ra{p}{v}$ at $t_1$. Queries, denoted Q : $q$, are updated at each level by using the latest values of the children. For example, to update the query at $t_3$, it takes the current value of $1$ from $\ra{p}{t}$ at $t_3$, then adds the weight of $5$ from $\ra{p}{b}$ at $t_0$, and the weight of $2$ from $\ra{p}{v}$ at $t_1$, as per Algorithm~\ref{alg:subsum}.}\label{fig:merging-operations}
\end{figure*}

\subsection{Batched mixed operations algorithm}

We are interested in evaluating batches of operations from an operation-set on trees with a \simpleRC{} implementation. In particular, we prove Theorem~\ref{theorem:batch}.

\begin{proofsketch}[Proof sketch of Theorem~\ref{theorem:batch}]

The preprocessing just builds an RC tree on the source tree, and sets the values for each cluster based on the initial
values on the leaves.   This can be implemented with the Miller-Reif algorithm~\cite{miller1989parallel}, in the binary forking model~\cite{blelloch2020optimal},
or deterministically~\cite{gazit1988optimal}.    All take linear work and logarithmic depth (w.h.p for the randomized versions).
Our algorithm for each batch is then implemented as follows:

\begin{enumerate}[leftmargin=*]
\setlength\itemsep{0em}
\item Timestamp the operations by their order in the sequence.
\item Collect all operations by their associated leaf, and sort within each leaf by timestamp.
This can be implemented with a single sort on the leaf identifier and timestamp.
\item For each leaf use a prefix sum on the update values to calculate the value of the leaf after each operation, starting from the initial value on the leaf.  
\item Initialize each query using the value it received from the prefix sum.
We now have a list of operations on each leaf sorted by timestamp.   For each update we have its value, and for each query we also have its partial evaluation based on the value. We prepend the initial value to the list, and call this the \emph{operation list}.
An operation list is \emph{non-trivial} if it has more than just the initial value.
\item For each level of the RC tree starting one above the deepest, and in parallel
for every cluster on the level for which at least one child has a non-trivial operation list:
\begin{enumerate}[leftmargin=*]
\item Merge the operation lists from each child into a single list sorted by timestamp.
\item Calculate for each element in the merged operations list, the latest value of each child at or before the timestamp.   This can be implemented by prefix sums. 
\item For each list element, calculate the value at that timestamp from the child values collected in the previous step.
\item For queries, use the values and/or child values to update the query.
\end{enumerate}
\end{enumerate}
This algorithm needs to have children with non-trivial operation lists identify parents that need to be processed.  This can be implemented by keeping a list of all the clusters at a level with non-trivial operation lists left-to-right in level order.   When moving up a level, clusters that share the same parent can be combined. An illustration of the merging process is depicted in Figure~\ref{fig:merging-operations} using the operations from Algorithm~\ref{alg:subsum}.

We first consider why the algorithm is correct.  We assume by structural induction (over subtrees) that the operation lists contain the correct values for each timestamped operation in the list.  This is true at the leaves since we apply a prefix sum across the associative operation to calculate the value at each update.  For internal clusters, assuming the child clusters have correct operation lists (values for each timestamp valid until the next timestamp, and partial result of queries), we properly determine the operation lists for the cluster.  In particular for all timestamps that appear in children we promote them to the parent, and for each we calculate the value based on the current value, by timestamp, for each child.

We now consider the costs.  The cost of the batch before processing the levels is dominated by the sort which takes $O(k \log k)$ work and $O(\log k)$ depth.  The cost at each level is then dominated by the merging and prefix sums which take $O(k)$ work and $O(\log k)$ depth accumulated across all clusters that have a child with a non-trivial operation list.  If the RC tree has depth $O(\log n)$ then across all levels the cost is bounded by $O(k \log n)$ work and $O(\log n\log k )$ depth.  The total work and depth is therefore as stated.  The space for each batch of size $k$ is bounded by the size of the RC tree which is $O(n)$ and the total space of the operation lists at any two adjacent levels, which is $O(k)$.
\end{proofsketch}


\subsection{Path updates and path/subtree queries}

We now consider implementing mixed operations consisting of updating paths, and querying both paths and subtrees.
We will use these in Sections~\ref{sec:improvements} and~\ref{sec:two-respect}.     
In particular
we wish to maintain, given a weighted rooted tree $T = (V,E)$, a data structure that supports the following operations.

\begin{itemize}[leftmargin=*]
  \setlength\itemsep{0em}
  \item \addPath($u, v, w$):  For $u,v \in V$ adds $w$ to the weight of all edges on
  the $u$ to $v$ path.
  \item \querySubtree($v$): Returns the lightest weight 
    of an edge in the subtree rooted at $v \in V$,
  \item \queryPath($u,v$): For $u, v \in V$, returns the lightest weight 
    of an edge on the $u$ to $v$ path.
  \item \queryEdge($e$): Returns $w(e)$
\end{itemize}

\noindent To implement these, we first implement the simpler operations \addPath'($v, w$), which adds weight $w$ to the path from $v$ to the root; and \queryPath'($u,v$), which requires that $v$ be the representative vertex of an ancestor of $u$ in the RC tree. The more general forms can be implemented in terms of these with a constant number of calls given the lowest common ancestor (LCA) in the original tree for
\addPath{} and in the RC tree for \queryPath{}.

\begin{figure}[t]
  \centering
  \includegraphics[width=0.4\columnwidth]{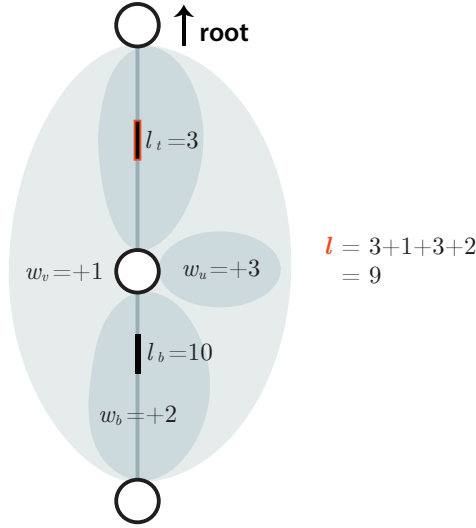}
  \caption{When a binary cluster joins its children, all \textproc{addPath}s' that originated in the vertex, bottom, or unary subclusters will affect all of the edges in the top cluster path. Here, $w' = w_v + w_b + w_u = 6$ weight is added to edges on the top cluster path due to \textproc{addPath} operations from below.}\label{fig:addpath}
\end{figure}

\begin{lemma}\label{lem:batch-add-path-subtree-query}
\label{lemma:path}
The \addPath', \querySubtree, \queryPath', and \queryEdge{} operations on bounded degree trees can be supported with a \simpleRC{} implementation.
\end{lemma}

  \begin{algorithm}
    \footnotesize
    \begin{algorithmic}[1]
      \State \textbf{using} \textbf{VertexV} = \textbf{int}
      \State \textbf{using} \textbf{UnaryV} = \textbf{struct} \{ $m$ : \textbf{edge}, $w$ : \textbf{int} \}
      \State \textbf{using} \textbf{BinaryV} = \textbf{struct} \{ $m$ : \textbf{edge}, $l$ : \textbf{edge}, $w$ : \textbf{int} \}\medskip
      \Procedure{$f_{\mbox{\small unary}}$}{$w_v$ : \textbf{VertexV}, $(m_t, l_t, w_t)$ : \textbf{BinaryV}, $U$ : \textbf{UnaryV list}}
          \State $w' \la w_v + \sum_{u \in U} \ra{u}{w}$
          \State $m_u \la \min_{u \in U} \ra{u}{m}$
          \State \algreturn{}  \{ $\min(m_t, l_t + w', m_u),  w_t + w'$ \}
      \EndProcedure 
      \Procedure{$f_{\mbox{\small binary}}$}{$w_v$ : \textbf{VertexV}, $(m_t, l_t, w_t)$ : \textbf{BinaryV}, $(m_b, l_b, w_b)$ : \textbf{BinaryV}, $U$ : \textbf{UnaryV list}}
          \State $w' \la w_v + w_b + \sum_{u \in U} \ra{u}{w}$
          \State $m_u \la \min_{u \in U} \ra{u}{m}$
          \State \algreturn{}  \{ $\min(m_t, m_b, m_u), \min(l_t + w', l_b), w_t + w'$ \}
      \EndProcedure
\Procedure{\addPath'}{$v$ : \textbf{vertex}, $w$ : \textbf{int}}
   \State $\ra{v}{\text{value}}$ \algassign{} $\ra{v}{\text{value}}$ + $w$
   \State Reevaluate the $f(\cdot)$ on path to root.
\EndProcedure
    \end{algorithmic}
    \caption{A \simpleRC{} implementation of \addPath{}'.}
    \label{alg:pathupdate}
  \end{algorithm}

\begin{algorithm}
\footnotesize
    \begin{algorithmic}[1]
\Procedure{\querySubtree}{$v$ : \textbf{vertex}}
  \State $w \la \infty$; $l \la \infty$
  \State $x \la v$;  $p \la \ra{x}{p}$
  \While{$p$ \textbf{is a} binary cluster}
    \If{$(x = \ra{p}{t})$ or $(x = \ra{p}{v})$}
      \State $w' \la \ra{\ra{p}{b}}{w} + \ra{\ra{p}{v}}{w} + \sum_{u \in \raa{p}{U}} \ra{u}{w}$
      \State $l \la \min(l + w', \ra{\ra{p}{b}}{l})$ \label{line:sublazy}
      \State $m \la \min(m, \ra{\ra{p}{b}}{m}, \min_{u \in \raa{p}{U}} \ra{u}{m})$ \label{line:submin}
    \EndIf
    \State $x \la p$; $p \la \ra{x}{p}$
  \EndWhile
  \State $w' \la \ra{\ra{p}{v}}{w} + \sum_{u \in \raa{p}{U}} \ra{u}{w}$
  \State return $\min(l + w', m, \min_{u \in \raa{p}{U}} \ra{u}{m})$
\EndProcedure

\Procedure{\queryEdge}{$e$ : \textbf{edge}}
  \State $w \la w(e)$
  \State $x \la e$; $p \la \ra{x}{p}$
  \While{$p$ \textbf{is a} binary cluster}
    \If{$x = \ra{p}{t}$}
      \State $w \la w + \ra{\ra{p}{b}}{w} + \ra{\ra{p}{v}}{w} + \sum_{u \in \raa{p}{U}} \ra{u}{w}$
    \EndIf
    \State $x \la p$; $p \la \ra{x}{p}$
  \EndWhile
  \State return $w + \ra{\ra{p}{v}}{w} + \sum_{u \in \raa{p}{U}} \ra{u}{w}$
\EndProcedure

\Procedure{\queryPath'}{$u$ : \textbf{vertex}, $v$ : \textbf{vertex}}
  \State $m \la \infty$; $t \la \infty$; $b \la \infty$
  \State $x \la u$;  $p \la \ra{x}{p}$
  \While{not $\ra{p}{v} = v$}\label{alg:querypathloopstart}
    \State $w' \la \ra{\ra{p}{v}}{w} + \sum_{u \in \raa{p}{U}} \ra{u}{w}$
    \If{$p$ \textbf{is a} unary cluster}
       \If{$x = \ra{p}{t}$} $m \la \min(t + w', m)$
       \Else~$m \la\min(\ra{\ra{p}{t}}{l} + w', m)$
       \EndIf
       \State $t \la \infty$; $b \la \infty$\label{alg:querypath-unary}
     \Else
       \State $w' \la w' + \ra{\ra{p}{b}}{w}$
       \If{$x = \ra{p}{t}$} $t \la t + w'$; $b \la \min(b + w', \ra{\ra{p}{b}}{l})$
       \ElsIf{$x = \ra{p}{b}$} $t \la \min(\ra{\ra{p}{t}}{l} + w', t)$
       \Else~$t \la \ra{\ra{p}{t}}{l} + w'$; $b \la \ra{\ra{p}{b}}{l}$
       \EndIf
     \EndIf
     \State $x \la p$; $p \la \ra{x}{p}$
  \EndWhile\label{alg:querypathloopend}
  \If{$x = \ra{p}{t}$} $l \la b$
  \ElsIf{$x = \ra{p}{b}$} $l \la t$
  \Else\ \textbf{return} $m$
  \EndIf
  \While{$p$ \textbf{is a} binary cluster}
    \State $w' \la \ra{\ra{p}{v}}{w} + \ra{\ra{p}{b}}{w} + \sum_{u \in \raa{p}{U}} \ra{u}{w}$
    \If{$(x = \ra{p}{t})$} $l \la l + w'$
    \EndIf
  \EndWhile        
  \State \textbf{return} $\min(m,l)$
\EndProcedure
    \end{algorithmic}
    \caption{A \simpleRC{} implementation of \queryEdge{}, \queryPath', and \querySubtree.}
    \label{alg:pathupdate2}
  \end{algorithm}

\begin{proofsketch}
  Our \simpleRC{} implementation for combining values and \addPath' is given in Algorithm~\ref{alg:pathupdate}. The queries are given in Algorithm~\ref{alg:pathupdate2}. The value of each vertex (leaf) in the cluster is the total weight added to that vertex by \addPath'.  The value for each unary cluster consists of: $m$, the minimum weight edge in the cluster; and $w$, the total weight of \addPath{}s' originating in the cluster.  For each binary cluster we separate the minimum weights on and off the cluster path.  In particular, the value of each binary cluster consists of: $m$, the minimum weight edge not on the cluster path; $l$, the minimum edge on the cluster path due to all \addPath{}' originating in the cluster; and $w$, the total weight of \addPath{}s' originating in the cluster.  The $f_{\mbox{\small binary}}$ and $f_{\mbox{\small unary}}$ calculate the values for unary and binary clusters from the values of their children.  We initialize each vertex with zero, and each edge $e$ with $(m = 0, l = w(e), w = 0)$.

  It is a \simpleRC{} implementation since (1) the $f(\cdot)$ can be computed in constant time, (2) the queries just traverse from a leaf on a path to the root (possibly ending early) only examining child values, taking constant time per level and constant space, and (3) the update just sets a leaf using an associative addition, and reevaluates the values to the root.
  
  We argue the implementation is correct.  Firstly we argue by structural induction on the RC tree that the values as described in the previous paragraph are maintained correctly by $f_{\mbox{\small binary}}$ and $f_{\mbox{\small unary}}$.  In particular assuming the children are correct we show the parent is correct.  The values are correct for leaves since we increment the value on vertices with \addPath', and initialize the edges appropriately.  To calculate the minimum edge weight of a unary cluster $f_{\mbox{\small unary}}$ takes the minimum of three quantities: the minimum off-path edge of the child binary cluster, the overall minimum edge of any of the child unary clusters, and, importantly, the minimum edge on the cluster path of the child binary cluster plus the \addPath{}' weight contributed by the unary clusters and the representative vertex (i.e., $\min(m_t, l_t + w', m_u)$).  This is correct since all paths from those clusters to the root go through the cluster path, so it needs to be adjusted.   The off-path edges and child unary clusters do not need to be adjusted since no path from the representative vertex goes through them.    The minimum weight is therefore correct.   The total \addPath{}' weight is correct since it just adds the contributions.
  
  For binary clusters we need to separately consider the minimum off- and on-path edges.     For the off-path edges
  the parts that are off the cluster path are the off-path edges from the two binary children, plus all edges from the unary children (i.e., $\min(m_t,m_b,m_u)$).    For the on-path edges both the top and bottom binary
  clusters contribute their on-path edges.    The on-path edges from the bottom binary cluster do not need to be adjusted because
  no vertices in the cluster are below them.    The on-path edges from the top binary cluster need to be adjusted by the
  \addPath{}' weights from all vertices
  in the bottom cluster, all vertices in unary child clusters, and the representative vertex since they are all below the path (this sum is given by $w'$). See Figure~\ref{fig:addpath}. The minimum of the resulted adjusted top edge and bottom edge is
  then returned, which is indeed the minimum edge on the path accounting for \addPath{}s' on vertices in the cluster.
  
  \querySubtree{}$(v)$ accumulates the appropriate minimum weights within a subtree as it goes up the RC tree. It starts at the node for which $v$ is its representative vertex. As with the calculation of values it needs to separate the on-path and off-path minimum weight.  Whenever coming as the upper binary cluster to the parent, \querySubtree{} needs to add all the contributing \addPath{}' weights from vertices below it in the parent cluster (the representative vertex, the lower binary cluster and the unary clusters, see Figure~\ref{fig:addpath}) to the current minimum on-path weight.  A minimum is then taken with the lower on-path minimum edge to calculate the new minimum on-path edge weight (Line~\ref{line:sublazy}). The off-path minimum is the minimum of the current off-path minimum, the minimum off-path edge of the bottom cluster and the minimums of the unary clusters (Line~\ref{line:submin}).  Once we reach a unary cluster we are done since for a unary
  cluster all subtrees of vertices within the cluster are fully contained within the cluster.    The final line therefore just determines the overal minimum for the subtree rooted at $v$ by considering the on-path edges adjusted by \addPath{}' contributions, the off-path edges, and all edges in child unary clusters.
  
  \queryEdge{}$(e)$ simply adds the total weight of all \addPath{}' operations that occurred beneath $e$ to the weight of $e$. Specifically, at each iteration of the loop, $w$ contains the $w(e)$ plus the total weight of all \addPath{}' operations originating at any vertex below $e$ that is contained in the current cluster $x$. As the query moves up the RC tree, if the parent cluster is a binary cluster and $x$ is its top subcluster, then the vertices not yet accounted for are those in the bottom subcluster, the representative vertex, and the unary subclusters. If $x$ is the bottom subcluster of its binary parent, or one of its unary subclusters, then no vertices in $p$ but not $x$ are below $e$. When the while loop terminates, $p$ is a unary cluster and $x$ is its binary subcluster. At this point, the representative of $p$, and all unary subclusters of $p$ are below $e$, and hence their weight is added to the total. Since $p$ is a unary cluster, there exists no additional vertices below $e$ in the tree, and hence the final weight contains the contributions of all \addPath{}' operations originating below $e$.
  
  Lastly, \queryPath{}' works by maintaining three values, $m, t, b$. To make defining them easier, consider, at each iteration of the main loop (Lines~\ref{alg:querypathloopstart}--\ref{alg:querypathloopend}) in which the current cluster $x$ is a binary cluster, the vertex $c$ which is the closest vertex to $u$ on the cluster path of $x$ (if $u$ is on the cluster path of $x$, say $c = u$). Then, we can define $m$ as the minimum weight edge on the path from $u$ to $c$ (which will be $\infty$ if $u$ is on the cluster path of $x$), $t$ as the minimum weight edge above $c$ on the cluster path of $x$, and $b$ as the minimum weight edge below $c$ on the cluster path of $x$. If $x$ is a unary cluster, then $t$ and $b$ are $\infty$ (undefined), and $m$ is simply the minimum weight edge on the path from $u$ to the boundary of $x$. Observe that it is important for the algorithm to maintain both $t$ and $b$ because it does not know in advance whether $v$ is above or below the current cluster path. It remains to argue that the implementation correctly maintains these values, and that the postprocessing is correct.
  
  Each time the algorithm moves up to the next highest cluster, it first computes $w'$, the total weight of all \addPath{}' operations originating below the representative vertex. If the cluster is a unary cluster, and $u$ originated from the top (binary) subcluster, then the path from $u$ to the boundary of $p$ consists of the previous path from $u$ to $c$ (the lightest edge on which is $m$), and the path from $c$ to the boundary of $p$ (the lightest edge on which is $t$). Since $w'$ weight has been added to all edges on the path from $c$ to the boundary of $p$, the lightest such edge is now $t + w'$ and hence the lightest edge on the path from $u$ to the boundary of $p$ is $\min(t + w', m)$. If $u$ did not originate in the top subcluster of $p$, it came from one of the unary subclusters. In this case, the path from $u$ to the boundary of $p$ consists of the path from $u$ to the boundary of $x$, and the cluster path of the top subcluster (which begins at the boundary of $x$ and ends at the boundary of $p$), and hence the lightest edge is $\min(p.t.l + w', m)$. Since the current cluster is a unary cluster, $t$ and $b$ are undefined (Line~\ref{alg:querypath-unary}).
  
  If the next cluster is a binary cluster, we reason as follows. If $u$ originated in the top subcluster, then the path from $c$ to the top boundary remains the same, but $w'$ weight is added to every edge (including $t$). The cluster path below $c$ now consists of the edges previously below $c$ to the bottom boundary of $x$, and additionally those on the cluster path of the bottom subcluster (the edges from the bottom boundary of $x$ to the bottom boundary of $p$). The edges below $c$ on the cluster path of the top subcluster (including $b$) have had their weight increased by $w'$, and hence the lightest edge on the path from $c$ to the bottom boundary of $p$ is now $\min(b + w', p.b.l)$. Similarly, if $u$ originated in the bottom subcluster, then the path from $c$ to the bottom boundary hasn't changed, so $b$ is unchanged, and no weight is added to the edge $t$. However, since the path from $c$ to the top boundary of $p$ now includes the cluster path of the top subcluster, the lightest edge from $c$ to the top boundary is now $\min(p.t.l + w', t)$. Otherwise, $u$ must have originated from a unary subcluster of $p$, and hence the cluster path of $p$ contains no edges from $x$, so $t$ is simply the lightest edge in the top subcluster, and $b$ is the lightest edge in the bottom subcluster.
  
  Once the main loop terminates (Lines~\ref{alg:querypathloopstart}--\ref{alg:querypathloopend}), by the loop condition, it must be because the current cluster $x$ has $v$ as a boundary. If $u$ originated in the top subcluster of the latest $p$, then $v$ must be the bottom boundary of $p$, and hence the path from $u$ to $v$ consists of the path from $u$ to $c$ and the path from $c$ to $v$ which goes towards the bottom boundary of $p$ and hence contains $b$. Conversely, if $u$ originated in the bottom subcluster of $p$, then the path from $u$ to $v$ goes towards the top boundary of $p$ and hence contains $t$. If $u$ originated in a unary subcluster, then the path from $u$ to $v$ just joins $u$ to the boundary of $x$, hence the lightest edge is $m$. If not, the lightest edge is either $m$, or $b$ or $t$ respectively. The weight of $b$ or $t$ might still be affected by \addPath{}' operations from below, so the total weight of such operations is accumulated by continuing up the RC tree and added to determine the final weight.
\end{proofsketch}

\begin{corollary} 
\label{corollary:batch}
Given a bounded-degree tree of size $n$, any sequence of $k$ \addPath, \querySubtree, \queryPath, and \queryEdge{} operations can be evaluated in $O(n + k \log (nk))$ work, $O(\log n \log k)$ depth and $O(n + k)$ space.  
\end{corollary}
\begin{proof}
The LCAs required to convert \addPath{} to \addPath'{} and \queryPath{} to \queryPath'{} can be computed in $O(n + m)$ work, $O(\log n)$ depth, and $O(n)$ space~\cite{schieber1988finding}.
The rest follows from Theorem~\ref{theorem:batch} and Lemma~\ref{lemma:path}.
\end{proof}

\subsection{Improving previous results}
\label{sec:improvements}

Using our batched mixed operations on trees algorithm, we can improve previous results on finding $2$-respecting cuts.  In particular we can
shave off a log factor in the work of Geissmann and Gianinazzi's parallel algorithm~\cite{geissmann2018parallel}, and we can parallelise Lovett and Sandlund's sequential algorithm~\cite{lovett2019simple}.

Geissmann and Gianinazzi find $2$-respecting cuts by first finding an $O(m)$ sequence of mixed \addPath{} and \queryPath{} operations for each of $O(\log n)$ trees.  They show how to find each sequence in $O(m \log n)$ work and $O(\log n)$ depth
On each set they then use their own data structure to evaluate the sequence in $O(m \log^2 n)$ work and $O(\log^2 n)$ depth, for a total of $O(m \log^3 n)$ work and $O(\log^2 n)$ depth across the sets.  Replacing their data structure with the result of Corollary~\ref{corollary:batch} improves their results to $O(m \log^2 n)$ work.

Lovett and Sandlund significantly simplify Karger's algorithm by first finding a heavy-light decomposition---i.e., a vertex disjoint set of paths in a tree such that every path in the tree is covered by at most $O(\log n)$ of them.   It then reduces finding the $2$-respecting cuts to a sequence of \addPath{} and \queryPath{} operations on the decomposed paths induced
by each non-tree edge, for a total of $O(m \log n)$ operations.  Using Geissmann and Gianinazzi's $O(n \log n)$ work $O(\log^2 n)$ algorithm for finding a heavy-light decomposition~\cite[Lemma~7]{geissmann2018parallel}, and the result of Corollary~\ref{corollary:batch} again gives an $O(m \log^2 n)$ work, $O(\log^2 n)$ depth algorithm.


%% file: packing.tex
\section{Producing the Tree Packing}\label{sec:packing}
  
    We follow the general approach used by Karger to produce a set of $O(\log n)$ spanning trees such that w.h.p., the minimum cut $2$ respects at least one of them. We have to make several improvements to achieve our desired work and depth bounds. At a high level, Karger's algorithm works as follows.
    \begin{enumerate}[leftmargin=*]
      \setlength\itemsep{0em}
      \item Compute an $O(1)$-approximate minimum cut $c$
      \item Sample edges from the unweighted multigraph corresponding to the weighted graph $G$, where an edge with weight $w$ is represented as $w$ parallel edges, with probability $\Theta(\log n/c)$
      \item Use the tree packing algorithm of Plotkin~\cite{plotkin1995fast} to generate a packing of $O(\log n)$ trees
    \end{enumerate}

  \noindent In this section, we describe the tools required to parallelise this algorithm.

  \subsection{A parallel version}
    
  Step 2 is trivial to parallelize, as the sampling can be done independently in parallel. The sampling procedure produces an unweighted multigraph with $O(m \log n)$ edges, and takes $O(m \log^2 n)$ work and $O(\log n)$ depth.
  
  In Step 3, Plotkin's algorithm consists of $O(\log^2 n)$ minimum spanning tree (MST) computations on a weighting of the sampled graph, which has $O(m \log n)$ edges. Naively this would require $O(m \log^3 n)$ work, but we can use a trick of Gawrychowski et al.~\cite{gawrychowski2019minimum}. Since the sampled graph is a multigraph sampled from $m$ edges, each invocation of the MST algorithm only cares about the current lightest of each parallel edge, which can be maintained in $O(1)$ time since the weights of the selected edges change by a constant each iteration. Using Cole, Klein, and Tarjan's linear-work MST algorithm~\cite{cole1996finding} results in a total of $O(m \log^2 n)$ work in $O(\log^3 n)$ depth w.h.p.
  
  The only nontrivial part of parallelizing the tree production is actually Step 1, computing an $O(1)$-approximate minimum cut. In the sequential setting, Matula's algorithm~\cite{matula1993linear} can be used, which runs in linear time on unweighted graphs, and on weighted graphs in $O(m \log^2 n)$ time. To the best of our knowledge, the only known parallelization of Matula's algorithm is due to Karger and Motwani~\cite{karger1994derandomization}, but it takes ${O}(m^2/n)$ work. We show how to compute an approximate minimum cut in $O(m \log^2 n)$ work and $O(\log^3 n)$ depth, which allows us to prove the following.
  
  \begin{theorem}\label{thm:tree-packing}
      Given a weighted graph, in $O(m\log^2 n)$ work and $O(\log^3 n)$ depth, a set of $O(\log n)$ spanning trees can be produced such that the minimum cut 2-respects at least one of them w.h.p.
    \end{theorem}
  
  \noindent We achieve our bounds by improving Karger's algorithms and speeding up several of the components. We use the following combination of ideas, new and old.
  \begin{enumerate}[leftmargin=*]
    \setlength\itemsep{0em}
    \item We extend a $k$-approximation algorithm of Karger~\cite{karger1993global} to work in parallel, allowing us to produce a $\log n$-approximate minimum cut in low work and depth.
    \item We use a faster sampling technique for producing Karger's skeletons for weighted graphs. This is done by transforming the graph into a graph that maintains an approximate minimum cut but has edge weights each bounded by $O(m \log n)$, and then using binomial random variables to sample all of the multiedges of a particular edge at the same time, instead of separately. Subsampling is then used to sample the same graph with decreasing probabilities.
    \item We show that the parallel sparse $k$-certificate algorithm of Cheriyan, Kao, and Thurimella~\cite{cheriyan1993scan} for unweighted graphs can be modified to run on weighted graphs.
    \item We show that Karger and Motwani's parallelization of Matula's algorithm can be generalized to weighted graphs.
    \item We use the $\log n$-approximate minimum cut to allow the algorithm to make just $O(\log\log n)$ guesses of the minimum cut such that at least one of them is an $O(1)$ approximation.
  \end{enumerate}
  
  \subsection{Parallel $\log n$-approximate minimum cut}
  
  To compute an $O(1)$-approximate minimum cut, our first step is actually to compute a $\log n$-approximate minimum cut. We parallelize an algorithm of Karger for computing $k$-approximate minimum cuts that is efficient when $k = \Omega(\log n)$~\cite{karger1993global}.
  
  \myparagraph{Mixed incremental connectivity and component weight queries} The following ingredient is useful in parallelizing Karger's $k$-approximate minimum cut algorithm. We show that that the following operations have a \simpleRC{} implementation, and hence can be efficiently implemented. Given a vertex-weighted, undirected graph with given initial vertex weights, we wish to support:
    \begin{itemize}[leftmargin=*]
      \setlength\itemsep{0em}
      \item \textproc{SubtractWeight}($v$, $w$): Subtract weight $w$ from vertex $v$
      \item \textproc{JoinEdge}($e$): Mark the edge $e$ as ``joined''
      \item \textproc{QueryWeight}($v$): Return the weight of the connected component containing the vertex $v$, where the components are induced by the joined edges 
    \end{itemize}
    
    \begin{lemma}
      The \textproc{SubtractWeight}, \textproc{JoinEdge}, and \textproc{QueryWeight} operations can be supported with a simple RC implementation.
    \end{lemma}
    
    \begin{proofsketch}
     The values stored in the RC clusters are as follows. Vertices store their weight, and unary clusters store the weight of the component reachable via joined edges from the boundary vertex. A binary cluster is either \emph{joined}, meaning that its boundary vertices are connected by joined edges, in which case it stores a single value, the weight of the component reachable via joined edges from the boundaries, otherwise it is \emph{split}, in which case it stores a pair: the weight of the component reachable via joined edges from the top boundary, and the weight of the component reachable via joined edges from the bottom boundary. We provide pseudocode for the update operations for Illustration in Algorithm~\ref{alg:component-weight}.
    
    \begin{algorithm}
      \footnotesize
        \begin{algorithmic}[1]
          \Procedure{$f_{\mbox{\small unary}}$}{$v_v, t, U$}
            \If{t = $(t_v, b_v)$}
              \algreturn{} $t_v$
            \Else\ 
              \algreturn{} $v_v + t + \sum_{u_v \in U} u_v$
            \EndIf
          \EndProcedure 
          \Procedure{$f_{\mbox{\small binary}}$}{$v_v, t, b, U$}
            \If{$t = t_v$ and $b = b_v$}
              \State \algreturn{} $t_v + b_v + v_v + \sum_{u_v \in U} u_v$
            \ElsIf{$t = (t_{t_v}, t_{b_v})$ \algand{} $b = b_v$}
              \State \algreturn{} $(t_{t_v}, t_{b_v} + v_v + b_v + \sum_{u_v \in U} u_v)$
            \ElsIf{$t = t_v$ \algand{} $b = (b_{t_v}, b_{b_v})$}
              \State \algreturn{} $(t_v + v_v + b_{t_v} + \sum_{u_v \in U} u_v)$
            \ElsIf{$t = (t_{t_v}, t_{b_v})$ \algand{} $b = (b_{t_v}, b_{b_v})$}
              \State \algreturn{} $(t_{t_v}, b_{b_v})$
            \EndIf
          \EndProcedure
    \Procedure{SubtractWeight}{$v, w$}
       \State $\ra{v}{\text{value}}$ \algassign{} $\ra{v}{\text{value}}$ - $w$
       \State Reevaluate the $f(\cdot)$ on path to root.
    \EndProcedure
    \Procedure{JoinEdge}{$e$}
      \State $\ra{e}{\text{value}}$ \algassign{} 0
      \State Reevaluate the $f(\cdot)$ on path to root.
    \EndProcedure
        \end{algorithmic}
        \caption{A \simpleRC{} implementation of \textproc{SubtractWeight} and \textproc{JoinEdge}.}
        \label{alg:component-weight}
    \end{algorithm}
    
    \noindent The initial value of a vertex is its starting weight. The initial value of an edge is $(0, 0)$, indicating that it is split at the beginning. Note that $f_{\mbox{\small unary}}$ and $f_{\mbox{\small binary}}$ can be evaluated in constant time, and the structure of the updates involves setting the value at a leaf using an associative operation and re-evaluating the values of the ancestor clusters.
    
    We can argue that the values are correctly maintained by structural induction. First consider unary clusters. If the top subcluster is split, then the representative vertex and unary subclusters are not reachable via joined edges, and hence the only reachable component is the component reachable inside the top subcluster from its top boundary, whose weight is $t_v$. If the top subcluster is joined, then the representative vertex is reachable, which is by definition the boundary vertex of the unary subclusters, and hence the reachable component is the union of the reachable components of all of the subclusters, whose weight is as given. 
    
    For binary clusters, there are four possible cases, depending on whether the top and bottom subclusters are joined or not. If both are joined, then the representative and hence the boundary of all subclusters is reachable from both boundaries, and hence the cluster is joined and the reachable component is the union of the reachable components of the subclusters. If either subcluster is split, then the reachable component at the corresponding boundary is just the reachable component of the subcluster, whose weight is as given. Lastly, if one of the subclusters is not split, then the corresponding boundary can reach the representative vertex, and hence the reachable components of the unary subclusters, whose weights are as given.
    
    It remains to argue that we can implement \textproc{QueryWeight} with a simple RC implementation. Consider a vertex $v$ whose component weight is desired and consider the parent cluster $P$ of $v$, i.e.,\ the cluster of which $v$ is the representative. If $P$ has no binary subclusters that are joined, observe that $P$ must contain the entire component of $v$ induced by joined edges, since the only way for a component to exit a cluster is via a boundary which would have to be joined. Answering the query in this situation is therefore easy; the result is the sum of the weights of $v$, the unary subclusters of $P$, the bottom boundary weight of the top subcluster (if it exists), and the top bounary weight of the bottom subcluster (if it exists). Suppose instead that $P$ contains a binary subcluster that is joined to some boundary vertex $u \neq v$. Since the subcluster is joined, $u$ is in the same induced component as $v$, and hence \textproc{QueryWeight}($v$) has the same answer as \textproc{QueryWeight}($u$). By standard properties of RC trees, since $u$ is a boundary of $P$, we also know that the leaf cluster $u$ is the child of some ancestor of $P$. Since the root cluster has no binary subclusters, this process of jumping to joined boundaries must eventually discover a vertex that falls into the easy case, and since such a vertex $u$ is always the child of some ancestor is $P$, the algorithm only examines clusters that are on or are children of the root-to-$v$ path in the RC tree, and hence the algorithm is a simple RC implementation.
    \end{proofsketch}
    
     \noindent Invoking Theorem~\ref{theorem:batch}, we obtain the following useful corollary.
    
    \begin{corollary}\label{lem:batch-component-weight}
      Given a vertex-weighted undirected graph, a batch of $k$ \textproc{SubtractWeight}, \textproc{JoinEdge}, and \textproc{QueryWeight} operations can be evaluated in $O(k\log (kn))$ work and $O(\log n \log k)$ depth.
    \end{corollary}
  
  \myparagraph{Parallel $k$-approximate minimum cut}
  Karger describes an $O(m n^{2/k} \log n)$ time sequential algorithm for finding a cut in a weighted graph within a factor of $k$ of the optimal cut~\cite{karger1993global}.  
  It works by randomly selecting edges to contract with probability proportional to their weight until a single vertex remains, and keeping track of the component with smallest incident weight (not including internal edges) during the contraction.
  
  His analysis shows that in a weighted graph with minimum cut $c$, with probability $n^{-2/k}$, the component with minimum incident weight encountered during a single trial of the contraction algorithm corresponds to a cut of weight at most $kc$, and therefore, running $O(n^{2/k}\log n)$ trials yields a cut of size at most $kc$ w.h.p.
  
  Although Karger's contraction algorithm is easy to parallelize using a parallel minimum spanning tree algorithm, keeping track of the incident component weights is trickier. To overcome this problem, we show that we can use our batch component weight algorithm to simulate the sequential contraction process efficiently. With this tool, we can determine the minimum incident weight of a component as follows:
    \begin{enumerate}[leftmargin=*]
      \setlength\itemsep{0em}
      \item\label{step:mst-for-approx} Compute an MST with respect to the weighted random edge ordering, where a heavier weight indicates that an edge contracts later
      \item\label{step:heaviest-edge-for-approx} For each edge $(u,v) \in G$, determine the heaviest edge in the MST on the unique $(u,v)$ path
      \item\label{step:batch-component-weight} Construct a vertex-weighted tree from the MST, where the weights are the total
  incident weight on each vertex in $G$. For each edge $(u,v)$ in the MST in contraction order:
      \begin{itemize}[leftmargin=*]
        \item Determine the set of edges in $G$ such that $(u,v)$ is the heaviest edge on its MST path. For each such edge identified, \textproc{SubtractWeight} from each of its endpoints by the weight of the edge
        \item Perform \textproc{JoinEdge} on the edge $(u,v)$
        \item Perform \textproc{QueryWeight} on the vertex $u$
      \end{itemize}
    \end{enumerate}
    
  \noindent Observe that the weight of a component at the point in time when it is queried is precisely the total weight of incident
  edges (again, not including internal edges). Taking the minimum over the initial degrees and all query results therefore yields the desired answer.
  
  Karger shows how to parallelize picking the (weighted) random permutation of the edges with $O(m \log^2 n)$ work.  It can easily slightly modified to improve the bounds by a logarithmic factor as follows. The algorithm selects the edges by running a prefix sum over the edge weights.  Assuming a total weight of $W$, it then picks $m$ random integers up to $W$, and for each uses binary search on the result of the prefix sum to pick an edge.  This process, however, might end up picking only the heaviest edges.  Karger shows that by removing those edges the total weight $W$ decreases by a constant factor, with high probability.  Since the edges can be preprocessed to be polynomial in $n$ (see below), repeating for $\log n$ rounds the algorithm will select all edges in the appropriate weighted random order.  Each round takes $O(m \log n)$ work for a total of $O(m \log^2 n)$ work.  
  
  Replacing the binary search in Karger's algorithm with a sort of the random integers and merge into the the result of the prefix sum yields an $O(m \log n)$ work randomized algorithm.  In particular $m$ random numbers uniformly distributed over a range can be sorted in $O(m)$ work and $O(\log n)$ depth by first determining for each number which of $m$ evenly distributed buckets within the range it is in, then sorting by bucket using an integer sort~\cite{RR89} and finally sorting within buckets.
  
  Step~\ref{step:mst-for-approx} therefore takes $O(m\log n)$ work and $O(\log^2 n)$ depth to compute the random edge permutation, and $O(m)$ work and $O(\log n)$ depth to run a parallel MST algorithm~\cite{karger1995randomized}. Step~\ref{step:heaviest-edge-for-approx} takes $O(m \log n)$ work and $O(\log n)$ depth using RC trees~\cite{acar2005experimental,acar2020batch}, and Step~\ref{step:batch-component-weight} takes $O(m \log n)$ work and $O(\log^2 n)$ depth by Corollary~\ref{lem:batch-component-weight} and the fact that the algorithm performs a batch of $O(n)$ operations.
  By Karger's analysis, trying $O(n^{2/k}\log n)$ random contractions yields the following lemma, and setting $k = \log n$ gives our desired corollary.
  
  \begin{lemma}\label{lem:k-approx}
    For a weighted graph, a cut within a factor of $k$ of the minimum cut can be found w.h.p.\ in $O(m n^{2/k} \log^2 n)$ work and $O(\log^2 n)$ depth. 
  \end{lemma}
  
  \begin{corollary}\label{cor:logn-approx}
    For a weighted graph, a cut within a factor of $\log n$ of the minimum cut can be found w.h.p.\ in $O(m \log^2 n)$ work and $O(\log^2 n)$ depth. 
  \end{corollary}
  
  \subsection{Additional tools and lemmas}
  
  \myparagraph{Transformation to bounded edge weights} For our algorithm to be efficient, we require that the input graph has small integer weights. Karger~\cite{karger1995random} gives a transformation that ensures all edge weights of a graph are bounded by $O(n^5)$ without affecting the minimum cut by more than a a constant factor. For our algorithm $O(n^5)$ would be too big, so we design a different transformation that guarantees all edge weights are bounded by $O(m \log n)$, and only affects the weight of the minimum cut by a constant factor.
    
    \begin{lemma}\label{lem:low-weight-transform}
      There exists a transformation that, given an integer-weighted graph $G$, produces an integer-weighted graph $G'$ no larger than $G$, such that $G'$ has edge weights bounded by $O(m\log n)$, and the minimum cut of $G'$ corresponds to an $O(1)$-approximate minimum cut in $G$.
    \end{lemma}
    
    \begin{proof}
      Let $G$ be the input graph and suppose that the true value of the minimum cut is $c$. First, we use Corollary~\ref{cor:logn-approx} to obtain a $O(\log n)$-approximate minimum cut, whose value we denote by $\tilde{c}$ ($c \leq \tilde{c} \leq c \log n$). We can contract all edges of the graph with weight greater than $\tilde{c}$ since they can not appear in the minimum cut. Let $s = \tilde{c} / (2m \log n)$. We delete (not contract) all edges with weight less than $s$. Since there are at most $m$ edges in any cut, this at most affects the value of a cut by $sm = \tilde{c} / (2 \log n) \leq c/2$. Therefore the minimum cut in this graph is still a constant factor approximation to the minimum cut in $G$.
      
      Next, scale all remaining edge weights down by the factor $s$, rounding down. All edge weights are now integers in the range $[1, 2m \log n]$. This is the transformed graph $G'$. It remains to argue that the value of the minimum cut is a constant-factor approximation. First, note that the scaling process preserves the order of cut values, and hence the true minimum cut in $G$ has the same value in $G'$ as the minimum cut in $G'$. Consider any cut in $G'$, and scale the weights of the edges back up by a factor $s$. This introduces a rounding error of at most $s$ per edge. Since any cut has at most $m$ edges, the total rounding error is at most $sm \leq c/2$. Therefore the value of the minimum cut in $G'$ is a constant factor approximation to the value of the minimum cut in $G$.
    \end{proof}
    
    \noindent Lastly, observe that this transformation can easily be performed in parallel by using a work-efficient connected components algorithm to perform the edge contractions, as is standard (see e.g.~\cite{karger1996new}).
  
  \myparagraph{Sampling binomial random variables} It will be helpful in the next step to be able to efficiently sample binomial random variables. We will
    use the following results due to Farach-Colton et al.~\cite{farach2015exact}.
    
    \begin{lemma}[Farach-Colton et al.~\cite{farach2015exact}, Theorem 1]
    Given a positive integer $n$, one can sample a random variate from the
    binomial distribution $B(n, 1/2)$ in $O(1)$ time with probability $1 - 1/n^{\Omega(1)}$ and in
    expectation after $O(n^{1/2+\varepsilon})$-time preprocessing for any constant $\varepsilon > 0$, assuming
    that $O(\log n)$ bits can be operated on in $O(1)$ time. The preprocessing can be reused for any $n' = O(n)$
    \end{lemma}
    \noindent We can also use the following reduction to sample $B(n,p)$ for arbitrary $0 \leq p \leq 1$.
    \begin{lemma}[Farach-Colton et al.~\cite{farach2015exact}, Theorem 2]
      Given an algorithm that can draw a sample from $B(n', 1/2)$ in $O(f(n))$
      time with probability $1 - 1/n^{\Omega(1)}$ and in expectation for any $n' \leq n$, then drawing a sample from $B(n', p)$ for any real p can be done in $O(f (n)\log n)$ time with probability $1 - 1/n^{\Omega(n)}$ and in expectation, assuming each bit of p can be obtained in $O(1)$ time
    \end{lemma}
     \noindent We note, importantly, that the model used by Farach-Colton et al.\ assumes that random $\Theta(\log n)$-size words can be generated in constant time. Since we only assume that we can generate random bits in constant time, we will have to account for this with an extra $O(\log n)$ factor in the work where appropriate. Note that this does not negatively affect the depth since we can pre-generate as many random words as we anticipate needing, all in parallel at the beginning of our algorithm. Lastly, we also remark that although it might not be clear in their definition, the constants in the algorithm can be configured to control the constant in the $\Omega(1)$ term in the probability, and therefore their algorithms take $O(1)$ time and $O(\log n)$ time w.h.p.
    
     To make use of these results, we need to show that the preprocessing of Lemma~\ref{lem:sample-binom} can be parallelized. Thankfully, it is easy. The preprocessing phase consists of generating $n^{\varepsilon}$ alias tables of size $O(\sqrt{n \log n})$. H{\"u}bschle-Schneider and Sanders~\cite{hbschleschneider2019parallel} give a linear work, $O(\log n)$ depth parallel algorithm for building alias tables. Building all of them in parallel means we can perform the alias table preprocessing in $O(n^{1/2+\varepsilon})$ work and $O(\log n)$ depth. The last piece of preprocessing information that needs to be generated is a lookup table for decomposing any integer $n' = O(n)$ into a sum of a constant number of square numbers. This table construction is trivial to parallelize, and hence all preprocessing runs in $O(n^{1/2+\varepsilon})$ work and $O(\log n)$ depth.
     
     \begin{lemma}\label{lem:sample-binom}
      Given a positive integer $n$, after $O(n^{1/2+\varepsilon})$ work and $O(\log n)$ depth preprocessing, one can sample random variables from $B(n, 1/2)$ in $O(\log n)$ work w.h.p., and from $B(n,p)$ in $O(\log^2 n)$ work w.h.p. The preprocessing can be reused for any $n' = O(n)$.
     \end{lemma}

  \myparagraph{Subsampling $p$-skeletons}
    Karger defines the $p$-skeleton $G(p)$ of an unweighted graph $G$ as a copy of $G$ where each edge appears with probability $p$. A $p$-skeleton therefore has $O(pm)$ edges in expectation. For a weighted graph, the $p$-skeleton is defined as the $p$-skeleton of the corresponding unweighted multigraph in which an edge of weight $w$ is replaced by $w$ parallel multiedges. The $p$-skeleton of a weighted graph therefore has $O(pW)$ edges in expectation, where $W$ is the total weight in the graph. Karger gives an algorithm for generating a $p$-skeleton in $O(pW \log(m))$ work, which relies on performing $O(pW)$ independent random samples with probabilities proportional to the weight of each edge, each of which takes $O(\log(m))$ amortized time. In Karger's algorithm, given a guess of the minimum cut $c$, he computes $p$-skeletons for $p = \Theta(\log n/c)$. Since no edge of weight greater than $c$ can be contained in the minimum cut, all such edges can be contracted, leaving us with $W \leq mc$, so the skeleton has $O(m \log n)$ edges and takes $O(m \log^2 n)$ work to compute. Since our algorithm does not know the minimum cut $c$ yet, it uses guessing and doubling on $p$, and hence has to compute several $p$-skeletons, so $O(m \log^2 n)$ work is too slow. We overcome this problem using binomial random variables and subsampling.
    
    \begin{lemma}\label{lem:subsampling}
      Given a weighted graph $G$ with edge weights bounded by $m^{2 - \varepsilon}$, an initial sampling probability $p$ and an integer $k$, there exists an algorithm that can produce the skeleton graphs $G(p), G(p/2), ..., G(p/2^k)$ in $O(m\log^2 n + km\log n)$ work w.h.p.\ and $O(k\log n)$ depth.
    \end{lemma}
    
    \begin{proof}
    Begin by using Lemma~\ref{lem:sample-binom} and performing the required preprocessing for sampling binomial random variables from $B(m^{2-\varepsilon}, 1/2)$, which takes $O(m)$ work and $O(\log n)$ depth. To construct $G(p)$, for each edge $e$ in the graph, sample a binomial random variable $x \sim B(w(e),p)$. The skeleton then contains the edge $e$ with weight $x$ (conceptually, $x$ unweighted copies of the multiedge $e$). This results in the same distribution of graphs as if sampled using Karger's technique, and takes $O(m \log^2 n)$ work w.h.p.\ and $O(\log n)$ depth. For each additional skeleton $G(p')$ requested, subsample from the previous skeleton by drawing binomial random variables from $B(w_{G(2p')}(e), 1/2)$, which takes $O(m \log n)$ work w.h.p.\ and $O(\log n)$ depth. In total, to perform $k$ rounds of sampling, this takes $O(m \log^2 n + km \log n)$ work w.h.p.\ and $O(k \log n)$ depth.
    \end{proof}
      
    \noindent Using subsampling here is important, since otherwise it would cost $O(km\log^2 n)$ work to sample all of the desired skeleton graphs. Additionally, note that Lemma~\ref{lem:low-weight-transform} makes it easy to satisfy the requirement that all edge weights be bounded by $m^{2 - \varepsilon}$.
  
  \myparagraph{Parallel weighted sparse certificates} 
  A \emph{sparse $k$-connectivity certificate} of an unweighted graph $G = (V,E)$ is a graph $G' = (V, E' \subset E)$ with at most $O(kn)$ edges, such that every cut in $G$ of weight at most $k$ has the same weight in $G'$. Cheriyan, Kao, and Thurimella~\cite{cheriyan1993scan} introduce a parallel graph search called \emph{scan-first search}, which they show can be used to generate $k$-connectivity certificates of unweighted graphs. Here, we briefly note that the algorithm can easily be extended to handle weighted graphs. The scan-first search algorithm is implemented as follows
    
    \begin{algorithm}
      \footnotesize
      \caption{Scan-first search~\cite{cheriyan1993scan}}
      \label{alg:scan_first_search}
      \begin{algorithmic}[1]
        \Procedure{SFS}{$G = (V,E)$ : \algtype{Graph}, $r$ : \algtype{Vertex}}
          \State Find a spanning tree $T'$ rooted at $r$
          \State Find a preorder numbering to the vertices in $T'$
          \State For each vertex $v \in T'$ with $v \neq r$, let $b(v)$ denote the least neighbor of $v$ in preorder
          \State Let $T$ be the tree formed by $\{v, b(v)\}$ for all $v \neq r$
        \EndProcedure
      \end{algorithmic}
    \end{algorithm}
  
    \noindent Using a linear work, low depth spanning tree algorithm, scan-first search can easily be implemented in $O(m)$ work and $O(\log n)$ depth. Cheriyan, Kao, and Thurimella show that if $E_i$ are the edges in a scan-first search forest of the graph $G_{i-1} = (V, E \setminus (E_1 \cup ... E_{i-1}))$, then $E_1 \cup ... E_k$ is a sparse $k$-connectivity certificate. A sparse $k$-connectivity certificate can therefore be found in $O(km)$ work and $O(k \log n)$ depth by running scan-first search $k$ times.
    
    In the weighted setting, we treat an edge of weight $w$ as $w$ parallel unweighted multiedges. As always, this is only conceptual, the multigraph is never actually generated. To compute certificates in weighted graphs, we therefore use the following simple modification. After computing each scan-first search tree, instead of removing the edges present from $G$, simply lower their weight by one, and remove them only if their weight becomes zero. It is easy to see that this is equivalent to running the ordinary algorithm on the unweighted multigraph. We therefore have the following.
  
  \begin{lemma}\label{lem:k-certificate}
    A sparse $k$-connectivity certificate for a weighted, undirected graph can be found in $O(km)$ work and $O(k \log n)$ depth.
  \end{lemma}

  \myparagraph{Parallelizing Matula's algorithm}
  Matula~\cite{matula1993linear} gave a linear time sequential algorithm for $(2+\varepsilon)$-approximate edge connectivity (unweighted minimum cut). It is easy to extend to weighted graphs so that it runs in $O(m \log n \log W)$ time, where $W$ is the total weight of the graph. Using standard transformations to obtain polynomially bounded edge weights, this gives an $O(m \log^2 n)$ algorithm. Karger and Motwani~\cite{karger1994derandomization} gave a parallel version of Matula's unweighted algorithm that runs in ${O}(m^2/n)$ work. Essentially, their version of Matula's algorithm does the following steps as indicated in Algorithm~\ref{alg:matula}.
    
    \begin{algorithm}[H]
      \footnotesize
      \caption{Approximate minimum cut}
      \label{alg:matula}
      \begin{algorithmic}[1]
        \Procedure{Matula}{$G = (V,E)$ : \algtype{Graph} }
          \If{$|V| = 1$} \textbf{return} $\infty$ \EndIf
          \State \alglocal{} $d$ \algassign{} minimum degree in $G$
          \State \alglocal{} $k$ \algassign{} $d / (2 + \varepsilon)$
          \State \alglocal{} $C$ \algassign{} Compute a sparse $k$-certificate of $G$
          \State \alglocal{} $G'$ \algassign{} Contract all non-certificate edges of $E$
          \State \algreturn{} $\min(d, \textproc{Matula}(G'))$
        \EndProcedure
      \end{algorithmic}
    \end{algorithm}
    
    \noindent It can be shown that at each iteration, the size of the graph is reduced by a constant factor, and hence there are at most $O(\log n)$ iterations. Furthermore, the work performed at each step is geometrically decreasing, so the total work using the sparse certificate algorithm of Cheriyan, Kao, and Thurimella~\cite{cheriyan1993scan} is $O(dm)$ and the depth is $O(d \log^2 n)$, where $d$ is the minimum degree of $G$.
    
    Here, we give a slight modification to this algorithm that makes it work on weighted graphs in $O(dm \log(W/m))$ work and $O(d \log n \log W)$ depth, where $d$ is the minimum weighted degree of the graph. 
    To extend the algorithm to weighted graphs, we can replace the sparse certificate routine with our modified version for weighted graphs, and replace the computation of $d$ with the equivalent weighted degree. By interpreting an edge-weighted graph as a multigraph where each edge of weight $w$ corresponds to $w$ parallel multiedges, we can see that the algorithm is equivalent. To argue the cost bounds, note that like in the original algorithm where the size of the graph decreases by a constant factor each iteration, the total weight of the graph must decrease by a constant factor in each iteration. Because of this, it is no longer true that the work of each iteration is geometrically decreasing. Naively, this gives a work bound of $O(d m \log(W))$, but we can tighten this slightly as follows. Observe that after performing $\log(W/m)$ iterations, the total weight of the graph will have been reduced to $O(m)$, and hence, like in the sequential algorithm, the work must subsequently begin to decrease geometrically. Hence the total work can actually be bounded by $O(d m \log(W/m) + dm) = O(dm \log(W/m))$. We therefore have the following.
  
  \begin{lemma}\label{lem:parallel-matula}
    Given a weighted graph with minimum weighted-degree $d$ and total weight $W$, an $O(1)$-approximate minimum cut can be found in $O(dm\log(W/m))$ work and $O(d \log n \log W)$ depth.
  \end{lemma}
  
  \subsection{Parallel $O(1)$-approximate minimum cut}
  
  We have finally amassed the ingredients needed to produce a parallel $O(1)$-approximate minimum cut algorithm. Well, we need one more trick, unsurprisingly due to Karger. To produce the sampled skeleton graph, Karger's algorithm chooses the sampling probability inversely proportional to the weight of the minimum cut, which paradoxically is what we are trying to compute. This issue is solved by using guessing and doubling. The algorithm guesses the minimum cut and computes the resulting approximation. It can then use Karger's sampling theorem~(Theorem 6.3.1 and Lemma 6.3.2 of \cite{karger1995random}) to verify whether the guess was too high.
  
  \begin{lemma}[Karger~\cite{karger1995random}]
    Let $G$ be a graph with minimum cut $c$ and let $p = \Theta((\log n)/\varepsilon^2 c)$. Then w.h.p.\ the minimum cut in $G(p)$ has value in $(1 \pm \varepsilon)pc$.
  \end{lemma}
  
  \begin{lemma}[Karger~\cite{karger1995random}]\label{lem:sampled_cut_size}
    w.h.p., if $G(p)$ is constructed and has minimum cut $\hat{c} = \Theta((\log n)\varepsilon^2)$ for $\varepsilon \leq 1$, then the minimum cut $c$ in $G$ has value in $(1 \pm \varepsilon)\hat{c}/p$.
  \end{lemma}
  
  \noindent If the true minimum cut is $c$, then the correct sampling probability for Karger's algorithm is $p = \Theta((\log n)\varepsilon^2 c)$, which produces a skeleton cut of size $\hat{c} = \Theta((\log n)/\varepsilon^2)$ w.h.p. If the algorithm makes a guess $C > 2c$ with corresponding probability $P = \Theta((\log n) / \varepsilon^2 C)$, then Lemma~\ref{lem:sampled_cut_size} says that the minimum cut in the skeleton graph is less than $\hat{c}$ w.h.p. The algorithm can therefore double the guess for $P$ and try again, until the minimum cut in the skeleton is larger than $\hat{c}$, at which point we know that the $P$-skeleton approximates the minimum cut within a factor $\varepsilon$. To perform these steps efficiently, our algorithm does the following:
  \begin{enumerate}[leftmargin=*]
  \item Transform the graph using Lemma~\ref{lem:low-weight-transform} to ensure that all weights are bounded by $O(m \log n)$ while retaining an $O(1)$-approximate minimum cut in $O(m \log^2 n)$ work and $O(\log^2 n)$ depth.
  \item Use Corollary~\ref{cor:logn-approx} to compute a $\log n$-approximate minimum cut value $C$ in $O(m \log^2 n)$ work and $O(\log^2 n)$ depth.
  \item Sample the skeleton graphs $G(\log^2 n / C), G(\log^2 n / (2C)), ..., G(\log n / C)$ using Lemma~\ref{lem:subsampling}. This is $\log \log n \leq \log n$ skeletons, and hence this takes $O(m \log^2 n)$ work w.h.p.\ and $O(\log^2 n)$ depth.
  \item For each skeleton graph:
    \begin{itemize}[leftmargin=*]
    \item Compute a sparse $\Theta(\log n)$ certificate of the skeleton graph. This takes $O(m \log n)$ work and $O(\log^2 n)$ depth by Lemma~\ref{lem:k-certificate}.
    \item Compute an $O(1)$-approximate minimum cut in the $\Theta(\log n)$ certificate using Matula's algorithm (Lemma~\ref{lem:parallel-matula}). Since the certificate guarantees that the total weight is at most $O(n \log n)$ and hence that the minimum weighted degree is at most $O(\log n)$, this takes $O(m \log n \log \log n)$ work and $O(\log^2 n \log \log n)$ depth.
    \end{itemize}
  \end{enumerate}
  
  \noindent Since there are $O(\log \log n)$ skeleton graphs, the total work done by the final step is at most $O(m \log n (\log \log n)^2)$, which is at most $O(m \log^2 n)$, and the depth is $O(\log^3 n)$. The correctness of the algorithm follows from the sampling theorem (Lemma~\ref{lem:sampled_cut_size}) and Karger's discussion~\cite{karger1995random}. Finally, we can conclude the following result.
  
  \begin{lemma}
    Given a weighted, undirected graph, the weight of an $O(1)$-approximate minimum cut can be computed w.h.p.\ in $O(m\log^2 n)$ work and $O(\log^3 n)$ depth
  \end{lemma}

%% file: two-respecting.tex
  \section{Finding Minimum $2$-respecting Cuts}\label{sec:two-respect}
    
    We are given a connected, weighted, undirected graph $G = (V,E)$ and a spanning tree $T$. In this section, we will give an algorithm that finds the minimum $2$-respecting cut of $G$ with respect to $T$ in $O(m \log n)$ work and $O(\log^3 n)$ depth.
    
    Our algorithm, like those that came before it, finds the minimum $2$-respecting cut by considering two cases. We assume that the tree $T$ is rooted arbitrarily. In the first case, we assume that the two tree edges of the cut occur along the same root-to-leaf path, i.e.\ one is a descendant of the other. This is called the \emph{descendant edges} case. In the second case, we assume that the two edges do not occur along the same root-to-leaf path. This is the \emph{independent edges} case.

    Since we are going to use RC trees, we require that $G$ have bounded degree. Note that any arbitrary degree graph can easily be \emph{ternarized} by replacing high-degree vertices with cycles of infinite weight edges, resulting in a graph of maximum degree three with the same minimum cut, and only a constant-factor larger size in terms of edges, which our bounds depend on.

    \subsection{Descendant edges}
       
    We present our minimum $2$-respecting cut algorithm for the descendant edges case. Let $T$ be a spanning tree of a connected graph $G = (V,E)$ of degree at most three, and root $T$ at an arbitrary vertex of degree at most two. The rooted tree is therefore a binary tree.
    
    We use the following fact. For any tree edge $e \in T$, let $F_e$ denote the set of edges $(u,v) \in E$ (tree and non-tree)
    such that the $u$ to $v$ path in $T$ contains the edge $e$.
    Then the weight of the cut induced by a pair of edges $\{ e, e' \}$ in $T$ is
    given by
    \begin{equation*}
    w(F_e \Delta F_{e'}) = w(F_e) + w(F_{e'}) - 2w(F_e \cap F_{e'}),
    \end{equation*}
    where $\Delta$ denotes the symmetric difference between the two sets. For each tree edge $e$, our algorithm seeks the tree edge $e'$ that minimizes
    $w(F_e \Delta F_{e'})$, which is equivalent to minimizing
    \begin{equation}
    w(F_{e'}) - 2w(F_e \cap F_{e'}).
    \end{equation}
    To do so, it traverses $T$ from the root while maintaining weights on a tree data structure that satisfies the following invariant:
    \begin{invariant}[Current subtree invariant]\label{inv:descendant-edges-dt}
      When visiting $e = (u,v)$, for every edge $e' \in \text{Subtree}(v)$, the
      weight of $e'$ in the dynamic tree is $w(F_{e'}) - 2w(F_e \cap F_{e'})$
    \end{invariant}
    \noindent The initial weight of each edge $e$ is therefore $w(F_e)$. Maintaining this invariant as the algorithm traverses the tree can then be
    achieved with the following observation.
    When the traversal descends from an edge $p = (w,u)$ to a neighboring child edge $e = (u,v)$, 
    the following hold for all $e' \in \text{Subtree}(v)$:
    \begin{enumerate}[leftmargin=*]
      \setlength\itemsep{0em}
      \item $(F_e \cap F_{e'})
      \supseteq (F_{p} \cap F_{e'})$, since any path that goes through $p$ and
      $e'$ must pass through $e$.
      \item $(F_e \cap
      F_{e'}) \setminus (F_{p} \cap F_{e'})$ are the edges $(x,y)
      \in F_{e'}$ such that
      $e$ is a \emph{top edge} of the path $x-y$ in $T$ (i.e., $e$ is on the path from $x$ to $y$ in $T$, but
      the parent edge of $e$ is not).
    \end{enumerate}
    Therefore, to maintain the current subtree invariant, when the algorithm visits the edge $e$, it need only subtract twice the weight
    of all $x-y$ paths that contain $e$ as a top edge. This can
    be done efficiently by precomputing the sets of top edges.   There are
    at most two top edges for each path $x-y$, and they can be found
    from the LCA of $x$ and $y$ in $T$.   We need not consider tree edges
    since they will never appear in $F_{e'}$.

    By maintaining the aforementioned invariant, the solution follows by taking the minimum value of $w(F_e) + \textproc{QuerySubtree}(v)$ for all edges $e = (u,v)$ during the traversal. As described, this algorithm is entirely sequential, but it can be parallelized using our batched mixed operations on trees algorithm (Corollary~\ref{corollary:batch}).
    
    The operation sequence can be generated as follows. First, the weights $w(F_e)$ for each edge can be computed using the batched mixed operations algorithm (Corollary~\ref{corollary:batch}) where each edge $(u,v)$ of weight $w$ creates an \textproc{AddPath}($u,v,w$) operation, followed by a \textproc{QueryEdge}($e$) for every edge $e \in T$. This takes $O(m\log n)$ work and $O(\log^2 n)$ depth. The LCAs required to compute the sets of top edges can be computed using the parallel LCA algorithm of Schieber and Vishkin~\cite{schieber1988finding} in $O(m)$ work and $O(\log n)$ depth in total. By computing an Euler tour of the tree $T$ (an ordered sequence of visited edges) beginning at the root, the order in which to perform the tree operations can be deduced in $O(n)$ work and $O(\log n)$ depth. Each edge in the Euler tour generates an \textproc{AddPath} operation for each of its top edges, followed by a \textproc{QuerySubtree} operation. Note that each edge is visited twice during the Euler tour. The second visit corresponds to negating the \textproc{AddPath} operations from the first visit. The solution is then the minimum result of all of the \textproc{QuerySubtree} operations. Since there are a constant number of top edges per path, and $O(m)$ paths in total, the operation sequence has length $O(m)$. Using Corollary~\ref{corollary:batch}, we arrive at the following.
    
    \begin{theorem}\label{thm:dependent}
      Given a weighted, undirected graph $G$ and a rooted spanning tree $T$, the minimum $2$-respecting cut of $G$ with respect to $T$ such that one of the cut edges is a descendant of the other can be computed in in $O(m \log n)$ work and $O(\log^2 n)$ depth w.h.p.
    \end{theorem}
    
    \begin{figure*}[t]
        \centering
        \includegraphics[width=0.9\textwidth]{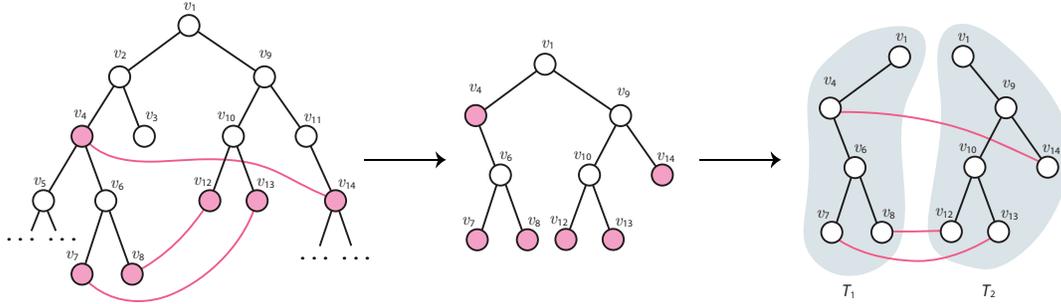}
        \caption{The bipartite problems are generated by compressing the input tree with respect to the endpoints of the edges whose endpoints share an LCA, then splitting the tree into the left and right halves.}\label{fig:bipartite}
      \end{figure*}
    
    \subsection{Independent edges}
    
    The independent edge case is where the two cutting edges do not fall on the same root-to-leaf path. To solve the independent edges problem, we use the framework of Gawrychowski et al.~\cite{gawrychowski2019minimum}, which is to decompose the problem into a set of subproblems, which they call \emph{bipartite problems}. The key challenge in parallelizing the solution to the bipartite problem is dealing with the fact that the resulting trees might not be balanced. The algorithm of Gawrychowski et al.\ relies on performing a biased divide-and-conquer search guided by a heavy-light decomposition~\cite{harel1984fast}, and then propagating results up the trees bottom up. Since the trees may be unbalanced, this can not be easily parallelized. Our solution is to use the recursive clustering of RC trees to guide a divide and conquer search in which we can maintain all of the needed information on the clusters.
    
    \begin{definition}[The bipartite problem]
      Given two weighted rooted trees $T_1$ and $T_2$ and a set of weighted edges that cross from one to the other, $L = \{ (u,v) :  u \in T_1, v \in T_2 \}$, the bipartite problem is to select $e_1 \in T_1$ and $e_2 \in T_2$ with the goal of minimizing the sum of the weight of $e_1$ and $e_2$ plus the weights of all edges $(v_1,v_2) \in L$ such that $v_1$ is in the subtree rooted at the bottom endpoint of $e_1$ and $v_2$ is in the subtree rooted at the bottom endpoint of $e_2$. The size of a bipartite problem is the size of $L$ plus the size of $T_1$ and $T_2$.
    \end{definition}
  
  \noindent Gawrychowski et al.\ observe that if $T_1$ and $T_2$ are edge-disjoint subtrees of $T$, then, assigning weights of $w(F_e)$ to each tree edge and weights of $-2w(e)$ to each non-tree edge, the solution to the bipartite problem is the minimum $2$-respecting cut such that $e_1 \in T_1$ and $e_2 \in T_2$. The independent edges problem is then solved by reducing it to several instances of the bipartite problem, and taking the minimum answer among all of them. We will show how to generate the bipartite problems efficiently, and how to solve them efficiently, both in parallel.

  \subsubsection{Generating the bipartite problems}
  
  The following parallel algorithm generates $O(n)$ instances of the bipartite problem with total size at most $O(m)$. For each edge $e$ in $T$, the algorithm first assigns them a weight equal to $w(F_e)$. Now consider all non-tree edges, i.e.\ all edges $e \in E, e \notin T$, group them by the LCA of their endpoints in $T$, and assign them a weight of $-2w(e)$. This forms a partition of the $O(m)$ edges of $G$, each group identified by a vertex. Each vertex in $T$ conversely has an associated (possibly empty) list of non-tree edges.
  
  For each vertex $v$ in $T$ with a non-empty associated list of edges, create a compressed path tree of $T$ with respect to the endpoints of the associated edges and $v$. Finally, for each such compressed path tree, root it at $v$ (the common LCA of the edge endpoints). The bipartite problems are now generated as follows. For each vertex $v$ with a non-empty list of non-tree edges, and the corresponding compressed path tree $T_v$, consider the children $x, y$ of $v$ in $T_v$. The bipartite problem consists of $T_1$, which contains the edge $(v,x)$ and the subtree of $T_v$ rooted at $x$, and likewise, $T_2$, which contains the edge $(v,y)$ and the subtree of $T_v$ rooted at $y$, and $L$, the associated list of non-tree edges. See Figure~\ref{fig:bipartite} for an illustration.
  
  \begin{lemma}
    Given a tree and a set of non-tree edges, the corresponding bipartite problems can be generated in $O(m \log n)$ work and $O(\log^2 n)$ depth w.h.p.
  \end{lemma}

  \begin{proof}
    The edge weight values can be computed in the same way as before using our batched mixed operations on trees algorithm in $O(m \log n)$ work and $O(\log^2 n)$ depth. LCAs can be computed using the parallel LCA algorithm of Schieber and Vishkin~\cite{schieber1988finding} in $O(m)$ work and $O(\log n)$ depth. Grouping the edges by LCA can be achieved using a parallel sorting algorithm in $O(m \log n)$ work and $O(\log n)$ depth. Together, these steps take $O(m \log n)$ work and $O(\log^2 n)$ depth. For each group, computing the compressed path tree takes $O(m_i \log(1+n/{m_i})) \leq O(m_i \log n)$ work and $O(\log^2 n)$ depth w.h.p., where $m_i$ is the number of edges in the group. Performing all compressed path tree computations in parallel and observing that the edge lists of each vertex are a disjoint partition of the edges of $G$, this takes at most $O(m \log n)$ work and $O(\log^2 n)$ depth in total w.h.p.
  \end{proof}

  \noindent It remains only for us to show that the bipartite problems can be efficiently solved in parallel.

  \subsubsection{Solving the bipartite problems}
  
   Our solution is a recursive algorithm that utilizes the recursive cluster structure of RC trees. Recall that RC trees consist of unary and binary clusters (and the nullary cluster at the root, but this is not needed by our algorithm).
   Since the bipartite problems are constructed such that trees $T_1$ and $T_2$ always have a root with a single child, the root cluster of their RC trees consists of exactly one unary cluster.

   \myparagraph{High-level idea} Recall that the goal is to select an edge $e_1 \in T_1$ and an edge $e_2 \in T_2$ that minimizes their costs plus the cost of all edges $(u,v) \in L$ such that $u$ is a descendant of $e_1$ and $v$ is a descendant of $e_2$. Our algorithm first constructs an RC tree of $T_1$, and weights the edges in $T_1$ and $T_2$ by their cost. At a high level, the algorithm then works as follows.   Given a binary cluster $c_1$ of $T_1$, the algorithm maintains weights on $T_2$ such that for each edge $e_2 \in T_2$, its weight is the weight of $e_2$ in the original tree plus the sum of the weights of all edges $(u,v) \in L$ such that $u$ is a descendant of the bottom boundary of $c_1$, and $v$ is a descendant of $e_2$.  This implies that for a binary cluster of $T_1$ consisting of an isolated edge $e_1 \in T_1$, the weights of each $e_2 \in T_2$ are precisely such that $w(e_1) + w(e_2)$ is the value of selecting $\{e_1, e_2\}$ as the solution.  This idea leads to a very natural recursive algorithm.  We start with the topmost unary cluster of $T_1$ and proceed recursively down the clusters of $T_1$, maintaining $T_2$ with weights as described. When the algorithm recurses into the top binary child of a cluster, it must add the weights of all $(u,v) \in L$ that are descendants of that cluster to the corresponding paths in $T_2$. If recursing on the bottom binary subcluster of a binary cluster, the weights on $T_2$ are unchanged. When recursing on a unary cluster, since it has no descendants, the algorithm uses the original weights of $T_2$. Once the recursion hits a binary cluster that consists of a single edge $e_1$, it can return the solution $w(e_1) + w(e_2)$, where $e_2$ is the lightest edge with respect to the current weights on $T_2$. Lastly, to perform this process efficiently, the algorithm \emph{compresses}, using the compressed path tree algorithm~\cite{anderson2020work}, the tree $T_2$ every time it recurses, keeping only the vertices that are endpoints of the crossing edges that touch the current cluster of $T_1$.

   \myparagraph{Implementation}    
   We provide pseudocode for our algorithm in Algorithm~\ref{alg:bipartite}. Given a bipartite problem $(T_1, T_2, L)$, we use the notation $L(C)$ to denote the edges of $L$ limited to those that are incident on some vertex in the cluster $C$. Furthermore, we use $V_{T_2}(L(C))$ to denote the set of vertices given by the endpoints of the edges in $L(C)$ that are in $T_2$. The pseudocode does not make the parallelism explicit, but all that is required is to run the recursive calls in parallel. 
The procedure takes as input a cluster $C$ of $T_1$, a compressed version of $T_2$ with its original weights, and $T_2'$, the compressed version of $T_2$ with updated weights. At the top level, it takes the cluster representing all of $T_1$ for the first argument, and the cluster for all of $T_2$ for the second and third argument. The \textproc{Compress} function compresses the given tree with respect to the given vertex set and its root, and returns the compressed tree still rooted at the same root. \textproc{AddPaths}($S$) takes a set $S \subset L$ of edges and for each one, adds $w(u,v)$ to the root-to-$v$ path, where $v \in T_2$, returning a new tree.

  \begin{algorithm}[h]
    \footnotesize
    \caption{Parallel bipartite problem algorithm}
    \label{alg:bipartite}
    \small
    \begin{algorithmic}[1]
      \small
      \Procedure{Bipartite}{$C$ : \algtype{Cluster}, $T_2$ : \algtype{Tree}, $T_2'$ : \algtype{Tree}, $L$ : \algtype{Edge list}}
        \If{$C = \{ e \}$}
          \State \algreturn{} $w(e) + \textproc{LightestEdge}(T_2')$
        \Else
          \State  $T_\text{cmp}$ \algassign{} $T_2$.\textproc{Compress}($V_{T_2}(L(\ra{C}{t}))$)
          \State  $T_2''$ \algassign{} $T_2'$.\textproc{AddPaths}($L(C) \setminus L(\ra{C}{t})$)
          \State  $T_\text{cmp}''$ \algassign{} $T_2''$.\textproc{Compress}($V_{T_2}(L(\ra{C}{t}))$)
          \State  ans \algassign{} \textproc{Bipartite}($\ra{C}{t}$, $T_\text{cmp}$, $T_\text{cmp}''$, $L(\ra{C}{t})$)
          \For{\algeach{} \algtype{cluster} $C'$ in $\ra{C}{U}$}
            \State  $T_\text{cmp}$ \algassign{} $T_2$.\textproc{Compress}($V_{T_2}(L(C'))$)
            \State ans \algassign{} $\min$(ans, \textproc{Bipartite}($C'$, $T_\text{cmp}$,  $T_\text{cmp}$, $L(C')$))
          \EndFor
          \If{$C$ \algis{} a \algtype{binary cluster}}
            \State  $T_\text{cmp}$ \algassign{} $T_2$.\textproc{Compress}($V_{T_2}(L(\ra{C}{b}))$)
            \State  $T_\text{cmp}'$ \algassign{} $T_2'$.\textproc{Compress}($V_{T_2}(L(\ra{C}{b}))$)
            \State ans \algassign{} $\min$(ans, \textproc{Bipartite}($T_\text{cmp}$, $T_\text{cmp}'$, $L(\ra{C}{b})$)) 
          \EndIf
          \State \algreturn{} ans
        \EndIf
      \EndProcedure
    \end{algorithmic}
  \end{algorithm}

  \noindent Since this algorithm creates many copies of $T_2$, we must ensure that we can still identify and locate a desired vertex given its label. One simple way to achieve this is to build a static hashtable alongside each copy of $T_2$ that maps vertex labels to the instance of that vertex in that copy.


  An ingredient that we need to achieve low depth is an efficient way to update the weights in $T_2$ when adding weights to a collection of paths. Although RC trees support batch-adding weights to paths, the standard algorithm does not meet our cost requirements. This is easy to achieve in linear work and $O(\log n)$ depth by propagating the total weight of all updates up the clusters, and then propagating back down the tree, the weight of all updates that are descendants of the current cluster. It remains to analyze the cost of the \textproc{Bipartite} procedure.
  
  \begin{theorem}
    A bipartite problem of size $m$ can be solved in $O(m \log m)$ work and $O(\log^3 m)$ depth w.h.p.
  \end{theorem}
  \begin{proof}
    First, since all recursive calls are made in parallel and the recursion is on the clusters of $T_1$, the number of levels of recursion is $O(\log m)$ w.h.p. We will show that the algorithm performs $O(m)$ work in total at each level, in $O(\log^2 m)$ depth w.h.p. Observe first that at each level of recursion, the edges $L$ for each call are a disjoint partition of the non-tree edges, since each recursive call takes a disjoint subset. We will now argue that each call does work proportional to $|L|$. Since $T_2$ and $T_2'$ are both compressed with respect to $L$, their size is proportional to $|L|$. \textproc{AddPaths} takes linear work in the size of $T_2$ and $O(\log m)$ depth, and hence takes $O(|L|)$ work and $O(\log m)$ depth. \textproc{Compress}($K$) takes $O(|K|\log(1+|T_2|/|K|)) \leq O(|K| + |T_2|)$ work and $O(\log^2 m)$ depth w.h.p.. Since compression is with respect to some subset of $L$, all of the compress operations take $O(|L|)$ work and $O(\log^2 m)$ depth w.h.p. In total, this is $O(|L|)$ work in $O(\log^2 m)$ depth w.h.p.\ at each level for each call. Since the $L$s at each level are a disjoint partition of the non-tree edges, the total work per level is $O(m)$ w.h.p., and hence the desired bounds follow.
  \end{proof}
  
  \noindent Since there are $O(n)$ bipartite problems of total size $O(m)$, solving them all in parallel yields the following, which, when combined with Theorem~\ref{thm:dependent}, proves Theorem~\ref{thm:two-respecting}. 
    
  \begin{theorem}
    Given a weighted, undirected graph $G$ and a rooted spanning tree $T$, the minimum $2$-respecting cut of $G$ with respect to $T$ such that the cut edges are independent can be computed in $O(m \log n)$ work and $O(\log^3 n)$ depth w.h.p.
  \end{theorem}
  
  \noindent Combining Theorem~\ref{thm:tree-packing} with Theorem~\ref{thm:two-respecting} on each of the $O(\log n)$ trees in parallel proves Theorem~\ref{thm:main-theorem}.

%% file: conclusion.tex
\section{Conclusion}
  
We present a randomized $O(m\log^2 n)$ work, $O(\log^3 n)$ depth parallel algorithm for minimum cut. It is the first parallel minimum cut algorithm to match the work bound of the best sequential algorithm, making it work efficient. Finding a faster parallel algorithm for minimum cut would therefore entail finding a faster sequential algorithm. It remains an open problem to find a deterministic algorithm for minimum cut, even a sequential one, that runs in $O(m \polylog n)$ time.
  

\section*{Acknowledgments}

We thank the anonymous referees for their comments and suggestions, and Phil Gibbons and Danny Sleator for their feedback on the manuscript. We thank Ticha Sethapakdi for helping with the figures. This research was supported by NSF grants CCF-1901381, CCF-1910030, and CCF-1919223.